\documentclass[submission,copyright,creativecommons]{eptcs}
\providecommand{\event}{ICE 2022} % Name of the event you are submitting to
\usepackage{underscore}           % Only needed if you use pdflatex.

\usepackage[dvipsnames]{xcolor}
\usepackage[normalem]{ulem}
\usepackage{tikz}
\usepackage[inline]{enumitem}

\usepackage{amsmath}
\usepackage{amssymb}
\usepackage{stmaryrd}
\usepackage{xspace}
\DeclareMathAlphabet{\mathcal}{OMS}{cmsy}{m}{n}

\usepackage{amsthm}
\usepackage{thmtools}
\usepackage{thm-restate}
\theoremstyle{plain}

\theoremstyle{definition}
\newtheorem{definition}{Definition}
\newtheorem{example}{Example}

\theoremstyle{remark}

\usepackage[capitalize]{cleveref}

\usepackage{xstring}

\usepackage{subcaption}

\usetikzlibrary{fit}
\usetikzlibrary{calc}
\pgfdeclarelayer{b1}
\pgfdeclarelayer{b2}
\pgfdeclarelayer{b3}
\pgfdeclarelayer{b4}
\pgfdeclarelayer{foreground}
\pgfsetlayers{b4,b3,b2,b1,main,foreground}
\tikzstyle{event} = [
  inner ysep = 0,
  inner xsep = 2pt
]
\newcommand{\pomsetbox}[2][dummy]{%
  \node[draw, inner ysep = 5pt, inner xsep = 3pt, rounded corners,% fill=white,
        fit = #2] (#1) {};%
}
\newcommand{\pomsetboxb}[3][1]{%
  \begin{pgfonlayer}{b#1}%
  \node[draw, inner ysep = 5pt, inner xsep = 3pt, rounded corners, fill=white,
        fit = #3] (#2) {};%
  \end{pgfonlayer}
}
\newcommand{\choicebox}[2][dummy]{
  \node[draw, inner sep = 5pt, %fill=blue!15,
        rounded corners, thick, fit = #2] (#1) {};
  \node[thick, fill = white, inner sep = 1pt, anchor = west, xshift = 6pt] at (#1.north west) {\tiny\textsc{Choice}};
}
\newcommand{\choiceboxb}[3][1]{
  \begin{pgfonlayer}{b#1}
  \node[draw, inner sep = 5pt, fill=blue!15,
        rounded corners, thick, fit = #3] (#2) {};
  \node[thick, fill = blue!15, inner sep = 1pt, anchor = west, xshift = 6pt, rounded corners = 2pt] at (#2.north west) {\tiny\textsc{Choice}};
  \end{pgfonlayer}
}
\definecolor{webgreen}{rgb}{0,.5,0}
\definecolor{webbrown}{rgb}{.6,0,0}
\definecolor{opcolor}{rgb}{0.4,0,0}

\newcommand{\nil}{\mathbf{0}}
\newcommand{\pt}[1]{\ensuremath{\textcolor{blue}{\mathsf{#1}}}\xspace}
\newcommand{\ms}[1]{\ensuremath{\textcolor{webgreen}{\texttt{\upshape{#1}}}}\xspace}
\newcommand{\subj}[1]{\ensuremath{\mathit{subj}(#1)}\xspace}

\newcommand{\isFinal}{{\downarrow}}
\newcommand{\act}[3]{\pt{#1}{}\textcolor{opcolor}{#2}\ms{#3}}

\newcommand{\alt}{\mathbin{\normalfont{\texttt{+}}}}
\newcommand{\seq}{\mathbin{;}}

\DeclareRobustCommand{\cod}[1]{
  \IfSubStr{#1}{->}{
    \pt{\StrBefore{#1}{-}}
    \textcolor{opcolor}{\shortrightarrow}
    \pt{\StrBetween{#1}{>}{:}}
    {:}
    \ms{\StrBehind{#1}{:}}
  }{
    \IfSubStr{#1}{?}{
      \pt{\StrBefore{#1}{?}}
      \textcolor{opcolor}{?}
      \ms{\StrBehind{#1}{?}}
    }{
      \IfSubStr{#1}{!}{
        \pt{\StrBefore{#1}{!}}
        \textcolor{opcolor}{!}
        \ms{\StrBehind{#1}{!}}
      }{
        \smurfsmurf
      }
    }
  }
}
\DeclareRobustCommand{\cods}[1]{\scriptstyle\cod{#1}}

\newcommand{\tr}[2][t]{\xrightarrow{\smash[#1]{#2}}}
\newcommand{\trpt}[2][t]{\tr[#1]{\checkmark_{\!\!#2}}} % partial termination
\newcommand{\trc}[2][t]{\tr[#1]{\checkmark_{\!\!#2}}} % checkmark, or something (for pomsets)
\newcommand{\set}[1]{\{#1\}}
\newcommand{\tpl}[1]{\langle #1 \rangle}
\newcommand{\sem}[1]{\llbracket #1 \rrbracket}

\newcommand{\dd}{\leq} % relation for direct dependencies

\newcommand{\Npom}{Branching pomset\xspace}
\newcommand{\Npoms}{Branching pomsets\xspace}
\newcommand{\npom}{branching pomset\xspace}
\newcommand{\npoms}{branching pomsets\xspace}

\newcommand{\A}{\mathcal{A}}
\newcommand{\X}{\mathcal{X}}
\newcommand{\N}{\mathcal{B}}
\newcommand{\C}{\mathcal{C}}
\renewcommand{\L}{\mathcal{L}}

\newcommand{\amin}{\text{a-min}}

\newcommand{\rulerefl}{\textsc{Refl}}
\newcommand{\ruletrans}{\textsc{Trans}}
\newcommand{\rulecongr}{\textsc{Congr}}
\newcommand{\rulechoice}{\textsc{Choice}}

\newcommand{\wrap}[1]{\begin{tabular}{@{}c@{}}#1\end{tabular}}
\newcommand{\mwrap}[1]{\begin{array}{@{}c@{}}#1\end{array}}

\newcommand{\bisim}{\mathcal{R}}

\title{
  Branching Pomsets for Choreographies
}

\author{
  Luc Edixhoven \qquad\qquad Sung-Shik Jongmans
  \institute{
    Open University (Heerlen) and\\CWI (Amsterdam), Netherlands}
  \email{\{led,ssj\}@ou.nl}
  \and
  Jos\'e Proen\c{c}a
  \institute{CISTER, ISEP,\\Polytechnic Institute of Porto, Portugal}
  \email{pro@isep.ipp.pt}
  \and
  Guillermina Cledou
  \institute{HASLab, INESC TEC \& University of Minho, Portugal}
  \email{mgc@inesctec.pt}
}

\def\authorrunning{Edixhoven, Jongmans, Proen\c{c}a, Cledou}

\def\titlerunning{Branching Pomsets for Choreographies}

\begin{document}

\maketitle

\begin{abstract}
  % \luc[Choreographies are cool]{Needs better wording}, yet their typical models (i.e., automata) are inherently sequential, which can result in an exponential number of states when modelling concurrency.
  % Pomsets are inherently concurrent and can model this concurrency in a compact way, but one may need an exponential number of pomsets to represent the branching of choices.
  % In this paper we present an augmented definition of pomsets:
  % we endow them with a hierarchical nesting structure to model branching in a single pomset.
  % The resulting model compactly represents both concurrency and choices.
  % As a proof of concept, we show how to derive such a pomset from a choreographic expression and we define well-formedness criteria on these pomsets which, when adhered to, guarantee realisability of the corresponding choreography in terms of bisimulation.

  %%%%
  % ------

  %Context
  Choreographic languages describe possible sequences of interactions among a set of agents. Typical models are based on languages or automata over sending and receiving actions. Pomsets provide a more compact alternative by using a partial order over these actions and by not making explicit the possible interleaving of concurrent actions.
   % in the presence of interleaving actions, .
  % used to reason over realisability of choreographies.
  %
  %Problem
  % {\color{red} However, choices can still produce an exponential number of states. For example, if an agent Alice can send 1 out of 3 possible messages to Bob \underline{twice}, a pomset (or an automata) to denote Alice's behaviour uses $3\times 3$ alternative pomsets (or states).}
  However, pomsets offer no compact representation of choices. For example, if an agent Alice can send one of two possible messages to Bob three times, one would need a set of $2 \times 2 \times 2$ distinct pomsets to represent all possible branches of Alice's behaviour.
  % Solution
  This paper proposes an extension of pomsets, named \emph{\npoms{}}, with a branching structure that can represent
  % {\color{red} Alice's behaviour using $3+3$ ordered actions}
  Alice's behaviour using $2 + 2 + 2$ ordered actions.
  % {\color{blue} choices explicitly and can thus represent the aforementioned behaviour using a total of $2n$ pairs of actions}.
  %
  %We define a semantics of choreographies using these \npoms{} and show that the pomset semantics are bisimilar to traditional operational semantics.
  We encode choreographies as \npoms and show that the pomset semantics of the encoded choreographies are bisimilar to their operational semantics.
  % We \luc[formalise the semantics of choreographies using \npoms{}]{Both SOS and pomsets?} and show that the choreography and the pomsets are bisimilar.
  % As a proof of concept, we propose %show how to define
  % well-formedness criteria on these pomsets which, when adhered to, guarantee
  % %realisability of the corresponding choreography with respect to bisimilarity.%
  % realisability of the corresponding choreographies, i.e., that the derived concurrent behaviour of the agents produces a bisimilar system to the original choreography.
\end{abstract}

\section{Introduction}
\label{sec:intro}

% \lucin{Before submitting: make all internal references consistent (Cref vs. cref).}
% \\\josein{Suggestion: use Cref in the beginning of the sentences, and use cref in the middle of the sentences.}

% \josein{Maybe mention the web-based tools with a graphical representation and animation of these pomsets, and implementing well-formedness and realisability via bisimilariy.}

Choreographic languages describe possible sequences of interactions, or communication protocols, \linebreak among a set of agents.
Their use is well established~\cite{ITU-T:MSC,DBLP:journals/tcs/AlurEY05,DBLP:conf/popl/HondaYC08,DBLP:journals/jacm/HondaYC16,DBLP:conf/popl/CarboneM13,DBLP:journals/tcs/Cruz-FilipeM20}, and it typically includes:

\begin{enumerate}
	\item reasoning statically over interaction properties;
	
	\item generating code that facilitates the implementation of the concurrent protocols.
\end{enumerate}

\noindent
Regarding 1,
static properties include deadlock absence or the equivalence between global protocols and the parallel composition of local protocols for each agent.
Regarding 2,
the code generated from choreographic languages include skeleton code for concurrent code, generated behavioural types that can be used to type-check agents, or dedicated orchestrators that dictate how the agents can interact.
In this work we focus on how to analyse choreographies by proposing a new structure to compactly represent their behaviour, based on \emph{partial-ordered multisets} (pomsets).

We foresee applications of this work in both aforementioned areas.
Regarding 1,
in static analysis, a more compact model of choreographies could reduce the complexity of the analysis of protocols featuring both concurrency and choices.
 Our work in progress in this area focuses on realisability.
Regarding 2,
our recent work includes API generation using an approach based on traditional sets of pomsets~\cite{DBLP:conf/ecoop/CledouEJP22}.
We are keen to extend it to take full advantage of the branching structure presented in this paper.

We use two simple running examples to motivate our approach.

\pagebreak

\begin{enumerate}
  \item[\textbf{1.}] \textbf{Master-workers (MW) protocol~\cite{DBLP:conf/cc/NeykovaY17}.}
  A \emph{master} (\pt{m}) concurrently sends \emph{tasks} (\ms{t}) to a number of \emph{workers} ($\pt{w_1}, \ldots, \pt{w_n}$).
  Once workers finish their task, they inform the master that they are \emph{done} (\ms{d}).
  This protocol is expressed in our choreographic language as follows for the case of two workers.
    $$(\cod{m->w_1:t} \seq \cod{w_1->m:d}) ~\parallel~ (\cod{m->w_2:t} \seq \cod{w_2->m:d}).$$
  Here, $\cod{m->w_1:t}$ represents an asynchronous communication from \pt{m} to \pt{w_1} of a message of type \ms{t}, `$\seq$' represents sequential composition and `$\parallel$' represents parallel composition.

  \item[\textbf{2.}] \textbf{Distributed voting (DV) protocol.}
  Three participants -- Alice ($\pt{a}$), Bob ($\pt{b}$) and Carol ($\pt{c}$) -- send their vote (\emph{yes} ($\ms{y}$) or \emph{no} ($\ms{n}$)) to every other participant in parallel. This is expressed as follows, where $\alt$ indicates nondeterministic choice.
  $$
  \left(\mwrap{(\cod{a->b:y} \parallel \cod{a->c:y}) \alt\\
               (\cod{a->b:n} \parallel \cod{a->c:n})~~\,}\right)
  \parallel
  \left(\mwrap{(\cod{b->a:y} \parallel \cod{b->c:y}) \alt\\
               (\cod{b->a:n} \parallel \cod{b->c:n})~~\,}\right)
  \parallel
  \left(\mwrap{(\cod{c->a:y} \parallel \cod{c->b:y}) \alt\\
               (\cod{c->a:n} \parallel \cod{c->b:n})~~\,}\right)
  $$
\end{enumerate}

% As an example, consider the master-workers protocol~\cite{DBLP:conf/cc/NeykovaY17}.
% In this protocol, a \emph{master} (\pt{m}) concurrently sends \emph{tasks} (\ms{t}) to some number of \emph{workers} ($\pt{w_1}, \ldots, \pt{w_n}$).
% When workers finish their task, they inform the master that they are \emph{done} (\ms{d}).
% We can express the master-workers protocol with two workers with the following choreography:
% $(\cod{m->w_1:t} \seq \cod{w_1->m:d}) \parallel (\cod{m->w_2:t} \seq \cod{w_2->m:d})$.
% Here, $\cod{m->w_1:t}$ represents an asynchronous communication from \pt{m} to \pt{w_1} of a message of type \ms{t}, `$\seq$' represents sequential composition and `$\parallel$' represents parallel composition.

% Another example is a distributed voting protocol in which every participant sends their vote (\emph{yes} ($\ms{y}$) or \emph{no} ($\ms{n}$)) to every other participant.
% We can express a distributed vote for three participants Alice ($\pt{a}$), Bob ($\pt{b}$) and Carol ($\pt{c}$) with the following choreography:
% $((\cod{a->b:y} \parallel \cod{a->c:y}) + (\cod{a->b:n} \parallel \cod{a->c:n}))
% \parallel
% ((\cod{b->a:y} \parallel \cod{b->c:y}) + (\cod{b->a:n} \parallel \cod{b->c:n}))
% \parallel
% ((\cod{c->a:y} \parallel \cod{c->b:y}) + (\cod{c->a:n} \parallel \cod{c->b:n}))$.

% We can express the distributed rock-paper-scissors protocol for two players Alice ($\pt{a}$) and Bob ($\pt{b}$) with the following choreography:
% $(\cod{a->b:r} \alt \cod{a->b:p} \alt \cod{a->b:s}) \parallel (\cod{b->a:r} \alt \cod{b->a:p} \alt \cod{b->a:s})$.

A protocol can evolve by performing sequences of sending and receiving actions. E.g., \act{ab}!x denotes a sending action from \pt{a} to \pt{b} with a message of type \ms{x}, and \act{ab}?x denotes the dual receiving action on \pt{b}.
Protocols with parallel interactions can have an explosion of states, such as our MW protocol, whose full state machine can be found on the left of \Cref{fig:automaton-vs-pomset}. To avoid this explosion, the state space can be represented more compactly using so-called \emph{partially ordered multisets}, or simply pomsets~\cite{DBLP:conf/fbt/KatoenL98,DBLP:journals/jlap/GuancialeT19}.
% (e.g., see the work of Katoen and Lambert~\cite{DBLP:conf/fbt/KatoenL98} and of Guanciale and Tuosto~\cite{GUANCIALE201969}).
The right of \Cref{fig:automaton-vs-pomset} shows a graphical pomset representation of the same MW protocol.
The pomset contains eight events, whose labels are shown.
The arrows visualise the partial order: an event precedes any other event to which it has an outgoing arrow, either directly or transitively.
In this example, the event with label $\cod{mw_1!t}$ precedes the event with label $\cod{mw_1?t}$ directly and the events with labels $\cod{w_1m!d}$ and $\cod{w_1m?d}$ transitively.
However, it is independent of the events involving $\pt{w_2}$.

The behaviour represented by a pomset is the set of all its linearisations, i.e., all sequences of the labels of its events that respect their partial order.
The set of linearisations of the pomset in \Cref{fig:automaton-vs-pomset} consists of all interleavings of the two threads $\cod{mw_1!t} \; \cod{mw_1?t} \; \cod{w_1m!d} \; \cod{w_1m?d}$ and $\cod{mw_2!t} \; \cod{mw_2?t} \; \cod{w_2m!d} \; \cod{w_2m?d}$.
This explicit concurrency yields a compact representation of the possible interleavings using just $4 + 4$ events, whereas the state machine needs $5 \times 5$ states to represent all interleavings.
% , where nodes represent actions that must be performed, and arrows denote a partial order among them.
% In this example the state machine needs $5 \times 5$ states to explicitly represent all possible interleavings of the actions of both workers. This interleaving is implicit in the pomset, which needs just $4 + 4$ events.
If we were to add a third worker, the automaton would grow by another factor 5, while the pomset would expand by just four additional events.

% In a recent paper~\cite{GUANCIALE201969}, Guanciale and Tuosto show the benefits of using partially ordered multisets (pomsets) as a model for choreographies.
% Where automata need an exponential number of states to represent concurrent behaviour, pomsets provide a compact way to represent parallelism.
% This is visualised in \Cref{fig:automaton-vs-pomset}: to represent a master-workers protocol with just two workers, an automaton needs $5 \times 5$ states to explicitly represent all possible interleavings of the actions of both workers, whereas the interleaving is implicit in the pomset, which needs just $4 + 4$ events.
% If we were to add a third worker, the automaton would grow by another factor 5, while the pomset would expand by just four additional events.

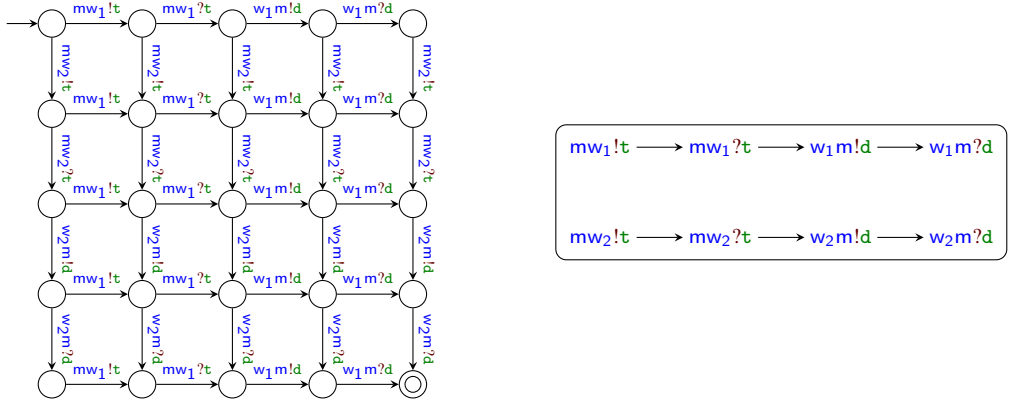
\begin{figure}[t]
  \hfill
  \wrap{% !TEX root = ../ice2022.tex
\begin{tikzpicture}[->, >=stealth, scale = .6, baseline = (40)]
  \tikzstyle{state} = [draw, circle, minimum height = 6pt]
  \tikzstyle{hor}=[above,inner sep=2pt]
  \tikzstyle{ver}=[hor,rotate=-90]

  \foreach \x in {0,1,2,3,4}
  \foreach \y in {0,1,2,3,4}
     \node[state] (\x\y) at (2*\x,-2*\y){};
  \draw (44) circle (5pt);
  \draw (-1,0) -- (00);

  \foreach \x in {0,1,2,3,4} {
    \draw (0\x) -- node[hor] {$\scriptscriptstyle\cod{mw_1!t}$} (1\x);
    \draw (1\x) -- node[hor] {$\scriptscriptstyle\cod{mw_1?t}$} (2\x);
    \draw (2\x) -- node[hor] {$\scriptscriptstyle\cod{w_1m!d}$} (3\x);
    \draw (3\x) -- node[hor] {$\scriptscriptstyle\cod{w_1m?d}$} (4\x);
    \draw (\x0) -- node[ver] {$\scriptscriptstyle\cod{mw_2!t}$} (\x1);
    \draw (\x1) -- node[ver] {$\scriptscriptstyle\cod{mw_2?t}$} (\x2);
    \draw (\x2) -- node[ver] {$\scriptscriptstyle\cod{w_2m!d}$} (\x3);
    \draw (\x3) -- node[ver] {$\scriptscriptstyle\cod{w_2m?d}$} (\x4);
   }
\end{tikzpicture}}
  \hfill
  \wrap{% !TEX root = ../ice2022.tex
\def\h{1.6}
\def\v{.6}
\begin{tikzpicture}[->, >=stealth]
  \node[event] (mw1!) at (0, \v) {$\cods{mw_1!t}$};
  \node[event] (mw1?) at (\h, \v) {$\cods{mw_1?t}$};
  \node[event] (w1m!) at (2*\h, \v) {$\cods{w_1m!d}$};
  \node[event] (w1m?) at (3*\h, \v) {$\cods{w_1m?d}$};
  \node[event] (mw2!) at (0, -\v) {$\cods{mw_2!t}$};
  \node[event] (mw2?) at (\h, -\v) {$\cods{mw_2?t}$};
  \node[event] (w2m!) at (2*\h, -\v) {$\cods{w_2m!d}$};
  \node[event] (w2m?) at (3*\h, -\v) {$\cods{w_2m?d}$};

  \draw (mw1!) -- (mw1?);
  \draw (mw1?) -- (w1m!);
  \draw (w1m!) -- (w1m?);
  \draw (mw2!) -- (mw2?);
  \draw (mw2?) -- (w2m!);
  \draw (w2m!) -- (w2m?);

  \pomsetbox{(mw1!) (w2m?)}
  % \node[pomsetbox, fit = (mw1!) (w2m?)] {};
\end{tikzpicture}}
  \hfill~

  \caption{An automaton (left) and a pomset (right) representing the master-workers protocol. % with two workers.
  \label{fig:automaton-vs-pomset}}
\end{figure}

While pomsets can compactly represent concurrent behaviour, \textbf{choices} need to be represented as \emph{sets} of pomsets: one for every branch.
As a consequence, one might need an exponential number of pomsets to represent a protocol with many choices.
The exponential growth is visible in our DV protocol with three participants, depicted on the left side of \Cref{fig:pomset-vs-branching-pomset}. This diagram represents a set of pomsets that capture the protocol's possible behaviour, counting $2 \times 2 \times 2$ different pomsets.
% to represent a three-participant distributed vote using a set of pomsets, although the individual pomsets are small we would need $2 \times 2 \times 2$ of them.
% to represent a two-player game of distributed rock-paper-scissors, although the individual pomsets are small, we would need $3 \times 3$ of them.
If we were to add a fourth participant, the set would grow by another factor 2.

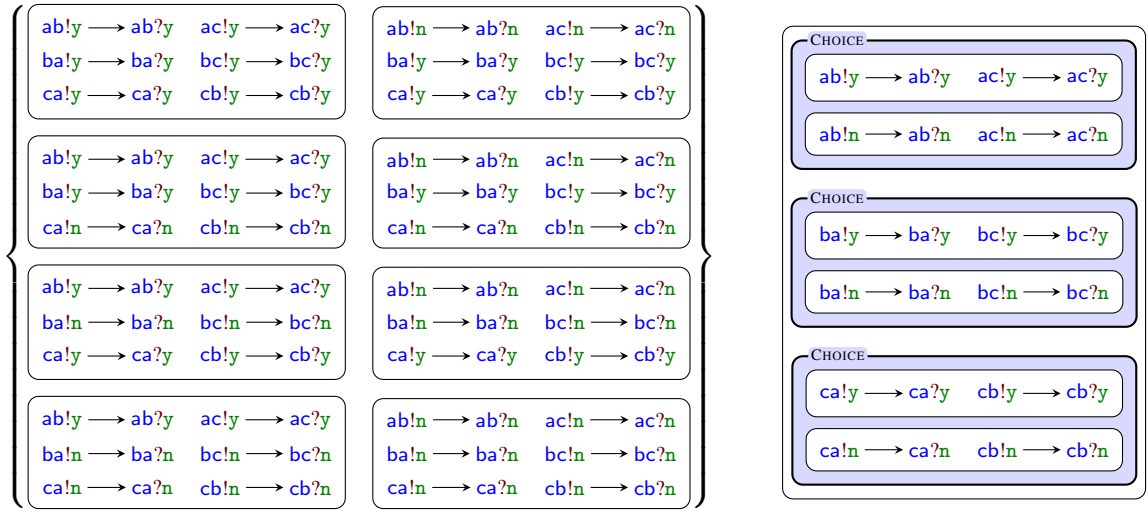
\begin{figure}[t]
  \hspace{\stretch{1}}
  $
  \def\v{0.6}
  \def\h{1.6}
  \def\hm{1.2} % mid distance
  \begin{Bmatrix}
    \begin{tikzpicture}[->, >=stealth, baseline = (anchor), scale = .75]
      \node (anchor) at (0,0) {};

      \node[event] (ab!) at (0,\v) {$\cods{ab!y}$};
      \node[event] (ab?) at (\h,\v) {$\cods{ab?y}$};
      \node[event] (ac!) at (\h+\hm,\v) {$\cods{ac!y}$};
      \node[event] (ac?) at (2*\h+\hm,\v) {$\cods{ac?y}$};
      \node[event] (ba!) at (0,0) {$\cods{ba!y}$};
      \node[event] (ba?) at (\h,0) {$\cods{ba?y}$};
      \node[event] (bc!) at (\h+\hm,0) {$\cods{bc!y}$};
      \node[event] (bc?) at (2*\h+\hm,0) {$\cods{bc?y}$};
      \node[event] (ca!) at (0,-\v) {$\cods{ca!y}$};
      \node[event] (ca?) at (\h,-\v) {$\cods{ca?y}$};
      \node[event] (cb!) at (\h+\hm,-\v) {$\cods{cb!y}$};
      \node[event] (cb?) at (2*\h+\hm,-\v) {$\cods{cb?y}$};

      \draw (ab!) -- (ab?);
      \draw (ac!) -- (ac?);
      \draw (ba!) -- (ba?);
      \draw (bc!) -- (bc?);
      \draw (ca!) -- (ca?);
      \draw (cb!) -- (cb?);

      \pomsetbox{(ab!) (cb?)}
    \end{tikzpicture}
    &
    \begin{tikzpicture}[->, >=stealth, baseline = (anchor), scale = .75]
      \node (anchor) at (0,0) {};

      \node[event] (ab!) at (0,\v) {$\cods{ab!n}$};
      \node[event] (ab?) at (\h,\v) {$\cods{ab?n}$};
      \node[event] (ac!) at (\h+\hm,\v) {$\cods{ac!n}$};
      \node[event] (ac?) at (2*\h+\hm,\v) {$\cods{ac?n}$};
      \node[event] (ba!) at (0,0) {$\cods{ba!y}$};
      \node[event] (ba?) at (\h,0) {$\cods{ba?y}$};
      \node[event] (bc!) at (\h+\hm,0) {$\cods{bc!y}$};
      \node[event] (bc?) at (2*\h+\hm,0) {$\cods{bc?y}$};
      \node[event] (ca!) at (0,-\v) {$\cods{ca!y}$};
      \node[event] (ca?) at (\h,-\v) {$\cods{ca?y}$};
      \node[event] (cb!) at (\h+\hm,-\v) {$\cods{cb!y}$};
      \node[event] (cb?) at (2*\h+\hm,-\v) {$\cods{cb?y}$};

      \draw (ab!) -- (ab?);
      \draw (ac!) -- (ac?);
      \draw (ba!) -- (ba?);
      \draw (bc!) -- (bc?);
      \draw (ca!) -- (ca?);
      \draw (cb!) -- (cb?);

      \pomsetbox{(ab!) (cb?)}
    \end{tikzpicture}
    \\
    \noalign{\medskip}
    % \cdots
    \begin{tikzpicture}[->, >=stealth, baseline = (anchor), scale = .75]
      \node (anchor) at (0,0) {};

      \node[event] (ab!) at (0,\v) {$\cods{ab!y}$};
      \node[event] (ab?) at (\h,\v) {$\cods{ab?y}$};
      \node[event] (ac!) at (\h+\hm,\v) {$\cods{ac!y}$};
      \node[event] (ac?) at (2*\h+\hm,\v) {$\cods{ac?y}$};
      \node[event] (ba!) at (0,0) {$\cods{ba!y}$};
      \node[event] (ba?) at (\h,0) {$\cods{ba?y}$};
      \node[event] (bc!) at (\h+\hm,0) {$\cods{bc!y}$};
      \node[event] (bc?) at (2*\h+\hm,0) {$\cods{bc?y}$};
      \node[event] (ca!) at (0,-\v) {$\cods{ca!n}$};
      \node[event] (ca?) at (\h,-\v) {$\cods{ca?n}$};
      \node[event] (cb!) at (\h+\hm,-\v) {$\cods{cb!n}$};
      \node[event] (cb?) at (2*\h+\hm,-\v) {$\cods{cb?n}$};

      \draw (ab!) -- (ab?);
      \draw (ac!) -- (ac?);
      \draw (ba!) -- (ba?);
      \draw (bc!) -- (bc?);
      \draw (ca!) -- (ca?);
      \draw (cb!) -- (cb?);

      \pomsetbox{(ab!) (cb?)}
    \end{tikzpicture}
    &
    \begin{tikzpicture}[->, >=stealth, baseline = (anchor), scale = .75]
      \node (anchor) at (0,0) {};

      \node[event] (ab!) at (0,\v) {$\cods{ab!n}$};
      \node[event] (ab?) at (\h,\v) {$\cods{ab?n}$};
      \node[event] (ac!) at (\h+\hm,\v) {$\cods{ac!n}$};
      \node[event] (ac?) at (2*\h+\hm,\v) {$\cods{ac?n}$};
      \node[event] (ba!) at (0,0) {$\cods{ba!y}$};
      \node[event] (ba?) at (\h,0) {$\cods{ba?y}$};
      \node[event] (bc!) at (\h+\hm,0) {$\cods{bc!y}$};
      \node[event] (bc?) at (2*\h+\hm,0) {$\cods{bc?y}$};
      \node[event] (ca!) at (0,-\v) {$\cods{ca!n}$};
      \node[event] (ca?) at (\h,-\v) {$\cods{ca?n}$};
      \node[event] (cb!) at (\h+\hm,-\v) {$\cods{cb!n}$};
      \node[event] (cb?) at (2*\h+\hm,-\v) {$\cods{cb?n}$};

      \draw (ab!) -- (ab?);
      \draw (ac!) -- (ac?);
      \draw (ba!) -- (ba?);
      \draw (bc!) -- (bc?);
      \draw (ca!) -- (ca?);
      \draw (cb!) -- (cb?);

      \pomsetbox{(ab!) (cb?)}
    \end{tikzpicture}
    \\
    \noalign{\medskip}
    \begin{tikzpicture}[->, >=stealth, baseline = (anchor), scale = .75]
      \node (anchor) at (0,0) {};

      \node[event] (ab!) at (0,\v) {$\cods{ab!y}$};
      \node[event] (ab?) at (\h,\v) {$\cods{ab?y}$};
      \node[event] (ac!) at (\h+\hm,\v) {$\cods{ac!y}$};
      \node[event] (ac?) at (2*\h+\hm,\v) {$\cods{ac?y}$};
      \node[event] (ba!) at (0,0) {$\cods{ba!n}$};
      \node[event] (ba?) at (\h,0) {$\cods{ba?n}$};
      \node[event] (bc!) at (\h+\hm,0) {$\cods{bc!n}$};
      \node[event] (bc?) at (2*\h+\hm,0) {$\cods{bc?n}$};
      \node[event] (ca!) at (0,-\v) {$\cods{ca!y}$};
      \node[event] (ca?) at (\h,-\v) {$\cods{ca?y}$};
      \node[event] (cb!) at (\h+\hm,-\v) {$\cods{cb!y}$};
      \node[event] (cb?) at (2*\h+\hm,-\v) {$\cods{cb?y}$};

      \draw (ab!) -- (ab?);
      \draw (ac!) -- (ac?);
      \draw (ba!) -- (ba?);
      \draw (bc!) -- (bc?);
      \draw (ca!) -- (ca?);
      \draw (cb!) -- (cb?);

      \pomsetbox{(ab!) (cb?)}
    \end{tikzpicture}
    &
    % \cdots
    \begin{tikzpicture}[->, >=stealth, baseline = (anchor), scale = .75]
      \node (anchor) at (0,0) {};

      \node[event] (ab!) at (0,\v) {$\cods{ab!n}$};
      \node[event] (ab?) at (\h,\v) {$\cods{ab?n}$};
      \node[event] (ac!) at (\h+\hm,\v) {$\cods{ac!n}$};
      \node[event] (ac?) at (2*\h+\hm,\v) {$\cods{ac?n}$};
      \node[event] (ba!) at (0,0) {$\cods{ba!n}$};
      \node[event] (ba?) at (\h,0) {$\cods{ba?n}$};
      \node[event] (bc!) at (\h+\hm,0) {$\cods{bc!n}$};
      \node[event] (bc?) at (2*\h+\hm,0) {$\cods{bc?n}$};
      \node[event] (ca!) at (0,-\v) {$\cods{ca!y}$};
      \node[event] (ca?) at (\h,-\v) {$\cods{ca?y}$};
      \node[event] (cb!) at (\h+\hm,-\v) {$\cods{cb!y}$};
      \node[event] (cb?) at (2*\h+\hm,-\v) {$\cods{cb?y}$};

      \draw (ab!) -- (ab?);
      \draw (ac!) -- (ac?);
      \draw (ba!) -- (ba?);
      \draw (bc!) -- (bc?);
      \draw (ca!) -- (ca?);
      \draw (cb!) -- (cb?);

      \pomsetbox{(ab!) (cb?)}
    \end{tikzpicture}
    \\
    \noalign{\medskip}
    % \vdots
    \begin{tikzpicture}[->, >=stealth, baseline = (anchor), scale = .75]
      \node (anchor) at (0,0) {};

      \node[event] (ab!) at (0,\v) {$\cods{ab!y}$};
      \node[event] (ab?) at (\h,\v) {$\cods{ab?y}$};
      \node[event] (ac!) at (\h+\hm,\v) {$\cods{ac!y}$};
      \node[event] (ac?) at (2*\h+\hm,\v) {$\cods{ac?y}$};
      \node[event] (ba!) at (0,0) {$\cods{ba!n}$};
      \node[event] (ba?) at (\h,0) {$\cods{ba?n}$};
      \node[event] (bc!) at (\h+\hm,0) {$\cods{bc!n}$};
      \node[event] (bc?) at (2*\h+\hm,0) {$\cods{bc?n}$};
      \node[event] (ca!) at (0,-\v) {$\cods{ca!n}$};
      \node[event] (ca?) at (\h,-\v) {$\cods{ca?n}$};
      \node[event] (cb!) at (\h+\hm,-\v) {$\cods{cb!n}$};
      \node[event] (cb?) at (2*\h+\hm,-\v) {$\cods{cb?n}$};

      \draw (ab!) -- (ab?);
      \draw (ac!) -- (ac?);
      \draw (ba!) -- (ba?);
      \draw (bc!) -- (bc?);
      \draw (ca!) -- (ca?);
      \draw (cb!) -- (cb?);

      \pomsetbox{(ab!) (cb?)}
    \end{tikzpicture}
    &
    % \vdots
    \begin{tikzpicture}[->, >=stealth, baseline = (anchor), scale = .75]
      \node (anchor) at (0,0) {};

      \node[event] (ab!) at (0,\v) {$\cods{ab!n}$};
      \node[event] (ab?) at (\h,\v) {$\cods{ab?n}$};
      \node[event] (ac!) at (\h+\hm,\v) {$\cods{ac!n}$};
      \node[event] (ac?) at (2*\h+\hm,\v) {$\cods{ac?n}$};
      \node[event] (ba!) at (0,0) {$\cods{ba!n}$};
      \node[event] (ba?) at (\h,0) {$\cods{ba?n}$};
      \node[event] (bc!) at (\h+\hm,0) {$\cods{bc!n}$};
      \node[event] (bc?) at (2*\h+\hm,0) {$\cods{bc?n}$};
      \node[event] (ca!) at (0,-\v) {$\cods{ca!n}$};
      \node[event] (ca?) at (\h,-\v) {$\cods{ca?n}$};
      \node[event] (cb!) at (\h+\hm,-\v) {$\cods{cb!n}$};
      \node[event] (cb?) at (2*\h+\hm,-\v) {$\cods{cb?n}$};

      \draw (ab!) -- (ab?);
      \draw (ac!) -- (ac?);
      \draw (ba!) -- (ba?);
      \draw (bc!) -- (bc?);
      \draw (ca!) -- (ca?);
      \draw (cb!) -- (cb?);

      \pomsetbox{(ab!) (cb?)}
    \end{tikzpicture}
  \end{Bmatrix}
  $
  \hspace{\stretch{1}}
  \begingroup
  \def\h{1.6} % horizontal distance
  \def\hm{1.2} % horizontal mid distance
  \def\v{1} % vertical distance
  \def\vb{1.8} % vertical distance between boxes
  \begin{tikzpicture}[scale = 0.75, ->, >=stealth, baseline = (anchor)]
    \node (anchor) at (0, 1.5*\v + \vb) {};

    \node[event] (ca!f) at (0, 0) {$\cods{ca!n}$};
    \node[event] (ca?f) at (\h, 0) {$\cods{ca?n}$};
    \node[event] (cb!f) at (\h + \hm, 0) {$\cods{cb!n}$};
    \node[event] (cb?f) at (2*\h + \hm, 0) {$\cods{cb?n}$};
    \node[event] (ca!t) at (0, \v) {$\cods{ca!y}$};
    \node[event] (ca?t) at (\h, \v) {$\cods{ca?y}$};
    \node[event] (cb!t) at (\h + \hm, \v) {$\cods{cb!y}$};
    \node[event] (cb?t) at (2*\h + \hm, \v) {$\cods{cb?y}$};

    \draw (ca!t) -- (ca?t);
    \draw (cb!t) -- (cb?t);
    \draw (ca!f) -- (ca?f);
    \draw (cb!f) -- (cb?f);

    \pomsetboxb{cf}{(ca!f) (cb?f)}
    \pomsetboxb{ct}{(ca!t) (cb?t)}

    \choiceboxb[2]{c}{(cf) (ct)}

    \node[event] (ba!f) at (0, \v + \vb) {$\cods{ba!n}$};
    \node[event] (ba?f) at (\h, \v + \vb) {$\cods{ba?n}$};
    \node[event] (bc!f) at (\h + \hm, \v + \vb) {$\cods{bc!n}$};
    \node[event] (bc?f) at (2*\h + \hm, \v + \vb) {$\cods{bc?n}$};
    \node[event] (ba!t) at (0, 2*\v + \vb) {$\cods{ba!y}$};
    \node[event] (ba?t) at (\h, 2*\v + \vb) {$\cods{ba?y}$};
    \node[event] (bc!t) at (\h + \hm, 2*\v + \vb) {$\cods{bc!y}$};
    \node[event] (bc?t) at (2*\h + \hm, 2*\v + \vb) {$\cods{bc?y}$};

    \draw (ba!t) -- (ba?t);
    \draw (bc!t) -- (bc?t);
    \draw (ba!f) -- (ba?f);
    \draw (bc!f) -- (bc?f);

    \pomsetboxb{bf}{(ba!f) (bc?f)}
    \pomsetboxb{bt}{(ba!t) (bc?t)}

    \choiceboxb[2]{b}{(bf) (bt)}

    \node[event] (ab!f) at (0, 2*\v + 2*\vb) {$\cods{ab!n}$};
    \node[event] (ab?f) at (\h, 2*\v + 2*\vb) {$\cods{ab?n}$};
    \node[event] (ac!f) at (\h + \hm, 2*\v + 2*\vb) {$\cods{ac!n}$};
    \node[event] (ac?f) at (2*\h + \hm, 2*\v + 2*\vb) {$\cods{ac?n}$};
    \node[event] (ab!t) at (0, 3*\v + 2*\vb) {$\cods{ab!y}$};
    \node[event] (ab?t) at (\h, 3*\v + 2*\vb) {$\cods{ab?y}$};
    \node[event] (ac!t) at (\h + \hm, 3*\v + 2*\vb) {$\cods{ac!y}$};
    \node[event] (ac?t) at (2*\h + \hm, 3*\v + 2*\vb) {$\cods{ac?y}$};

    \draw (ab!t) -- (ab?t);
    \draw (ac!t) -- (ac?t);
    \draw (ab!f) -- (ab?f);
    \draw (ac!f) -- (ac?f);

    \pomsetboxb{af}{(ab!f) (ac?f)}
    \pomsetboxb{at}{(ab!t) (ac?t)}

    \choiceboxb[2]{a}{(af) (at)}

    \pomsetbox{(a) (b) (c)}
  \end{tikzpicture}
  \endgroup
  \hspace{\stretch{1}}

  \caption{A set of pomsets (left) and a \npom{} (right) representing a three-participant distributed vote.}
  \label{fig:pomset-vs-branching-pomset}
\end{figure}

This paper proposes an extension to pomsets, named \emph{\npoms{}}, with a branching structure that can compactly represent choices.
A \npom{} initially contains all branches of choices, and discards non-chosen branches when firing events that require resolving a choice.
The right side of \Cref{fig:pomset-vs-branching-pomset} depicts an example of a \npom{} for our DV protocol: where we would traditionally need $2 \times 2 \times 2$ pomsets (with six pairs of events each), we can represent the same behaviour as a single \npom{} with $2 + 2 + 2$ choices (with four pairs of events each).
Adding an additional participant would double the number of pomsets in the set of pomsets, while it would add a single choice to the \npom{}.

To aid in the understanding of \npoms{} and their semantics, we provide a prototype tool to visualise them, available at \url{https://arca.di.uminho.pt/b-pomset/}.
The tool provides a web interface where one can submit an input choreography, which is then visualised as a \npom{} and can be simulated.
The examples and figures in the paper are already available as preset inputs.
We note that the pomset simulation in our prototype currently does not support loops, for reasons which will become apparent later in the paper; however, all other operators are supported and we are most interested in (combinations of) choice and parallel composition.

\paragraph{Contribution}
This paper provides three core contributions: (1) an extension of pomsets with a branching structure, named \npoms, (2) an encoding from a choreographic language into \npoms, and (3) a formal proof that the operational semantics of a choreography and of its encoded \npom are equivalent, i.e., bisimilar.

\paragraph{Structure of the paper}
% The rest of this paper is structured as follows.
\Cref{sec:choreography} presents the syntax of our choreography language and its operational semantics.
\Cref{sec:pomsets} formalises \npoms{} and their semantics.
\Cref{sec:chor-to-pom} formalises how to obtain a \npom{} from a choreography and shows that a choreography and its derived \npom{} are behaviourally equivalent.
% \Cref{sec:realisability} identifies well-formedness criteria on \npoms{} and shows that any well-formed pomset is realisable --- and, consequently, also any equivalent choreography.
Finally, \Cref{sec:conclusion} presents our conclusions and a brief discussion about future work and related work.

% \lucin{
%   Running example (non-academic, presumably): master-workers with slight changes to fit the situation (concurrency, choices).
% Suggestion for running example: $((\cod{a->b:x} \alt \cod{a->b:y}) \parallel (\cod{b->c:x} \alt \cod{b->d:x})) \seq \cod{a->c:x}$. No special story, but:
% \\
% - Parallelism to annoy automata
% \\
% - Four different branches to annoy traditional pomsets
% \\
% - Dependencies between nested and non-nested events (when interpreted as pomsets)
% \\
% - Weak sequential composition between $(\cod{a->b:x} \alt \cod{a->b:y})$ and $\cod{a->c:x}$
% \\
% - Partial termination and the pomset congruence rule with $(\cod{b->c:x} \alt \cod{b->d:x})$ and $\cod{a->c:x}$
% }

\section{Choreographies}
\label{sec:choreography}

In this section, we formally define the syntax and semantics of our choreographic language, examples of which have been shown in the previous section.

Let $\A$ be the set of all participants $\pt{a}, \pt{b}, \ldots$.
Let $\X$ be the set of all message types $\ms{x}, \ms{y}, \ldots$.
Let $\L = \bigcup_{\pt{a},\pt{b} \in \A, \ms{x} \in \X} \set{\cod{ab!x}, \cod{ab?x}}$ be the set of actions $\set{\cod{ab!x}, \cod{ab?x}}$ for all participants $\pt{a} \neq \pt{b}$ and message types $\ms{x}$.
For all actions the \emph{subject} of the action is its active participant: the subject of a send action $\cod{ab!x}$ is $\pt{a}$, written $\subj{\cod{ab!x}} = \pt{a}$, and the subject of a receive action $\cod{ab?x}$ is $\pt{b}$.

The syntax is formally defined in \Cref{fig:chor-syntax}.
Its components are standard:
`$\nil$' is the empty choreography;
`$\cod{a->b:x}$' is the asynchronous communication from $\pt{a}$ to $\pt{b}$ of a message of type $\ms{x}$;
the boxed term `$\cod{ab?x}$' represents a pending receive on $\pt{b}$ from $\pt{a}$ of a message of type $\ms{x}$ (it is boxed in \Cref{fig:chor-syntax} to indicate that it is only used internally to formalise behaviour but the box is not part of the syntax);
`$c_1 \seq c_2$', `$c_1 \alt c_2$' and `$c_1 \parallel c_2$' are respectively the weak sequential composition, nondeterministic choice and parallel composition of choreographies $c_1$ and $c_2$;
finally, `$c^*$' is the finite repetition (or, more informally, loop) of choreography $c$.
The semantics for choice, parallel composition and loop are standard.
We note that our sequential composition is weak.
More traditionally, when sequencing $c_1$ and $c_2$, the choreography $c_1$ must fully terminate before proceeding to $c_2$.
With weak sequential composition, however, actions in $c_2$ can already be executed as long as they do not interfere with $c_1$.
For example, in $\cod{a->b:x} \seq \cod{c->d:x}$ we can execute the action $\cod{cd!x}$ as it does not affect the participants of $\cod{a->b:x}$: there is no dependency and thus no need to wait for $\cod{a->b:x}$ to go first.
However, in $\cod{a->b:x} \seq \cod{a->c:x}$ the action $\cod{ac!x}$ cannot be executed first as its subject ($\pt{a}$) must first execute $\cod{ab!x}$.
This is the common interpretation of sequential composition in the context of message sequence charts~\cite{DBLP:conf/fbt/KatoenL98}, multiparty session types~\cite{DBLP:conf/popl/HondaYC08} and choreographic programming~\cite{DBLP:conf/popl/CarboneM13}.

\begin{figure}[t]
  \[
    c ::= \nil ~|~ \cod{a->b:x} ~|~ \smash{\boxed{\cod{ab?x}}} ~|~ c \seq c ~|~ c \alt c ~|~ c \parallel c ~|~ c^*
  \]

  \caption{Syntax of choreographies, where $\pt{a}$ and $\pt{b}$ are participants ($\pt{a} \neq \pt{b}$) and $\ms{x}$ is a message type.}
  \label{fig:chor-syntax}
\end{figure}

The reduction rules of our choreographic language are formally defined in \Cref{fig:chor-semantics-reduction} and its termination rules in \Cref{fig:chor-semantics-termination}.
To formalise the reduction of weak sequential composition, we follow Rensink and Wehrheim~\cite{DBLP:journals/acta/RensinkW01}, who define a notion of \emph{partial termination}.
% \jose{Maybe mention this paper in the related work section}

\begin{figure}[t]
  \begin{subfigure}{\textwidth}
    \begin{gather*}
    \dfrac{}{\cod{a->b:x} \tr[]{\cods{ab!x}} \cod{ab?x}}
    \quad
    \dfrac{}{\cod{ab?x} \tr[]{\cods{ab?x}} \nil}
    \quad
    \dfrac{c_1 \tr{\ell} c_1'}{c_1 \seq c_2 \tr[]{\ell} c_1' \seq c_2}
    \quad
    \dfrac{c_1 \trpt{\ell} c_1' \quad c_2 \tr{\ell} c_2'}{c_1 \seq c_2 \tr[]{\ell} c_1' \seq c_2'}
    \\[2mm]
    \dfrac{c_1 \tr[]{\ell} c_1'}{c_1 \parallel c_2 \tr[]{\ell} c_1' \parallel c_2}
    \quad
    \dfrac{c_2 \tr{\ell} c_2'}{c_1 \parallel c_2 \tr[]{\ell} c_1 \parallel c_2'}
    \quad
    \dfrac{c_1 \tr{\ell} c_1'}{c_1 \alt c_2 \tr[]{\ell} c_1'}
    \quad
    \dfrac{c_2 \tr{\ell} c_2'}{c_1 \alt c_2 \tr[]{\ell} c_2'}
    \quad
    \dfrac{c \tr{\ell} c'}{c^* \tr[]{\ell} c' \seq c^*}
    \end{gather*}

    \caption{Reduction rules.}
    \label{fig:chor-semantics-reduction}
  \end{subfigure}

  \begin{subfigure}{\textwidth}
    \begin{gather*}
    \dfrac{}{\nil \isFinal}
    \quad
    \dfrac{}{c^* \isFinal}
    \quad
    \dfrac{c_1 \isFinal \quad c_2 \isFinal \quad \dagger \in \set{{\seq}, {\parallel}}}{c_1 \dagger c_2 \isFinal}
    \quad
    \dfrac{c_i \isFinal \quad i \in \set{1, 2}}{c_1 \alt c_2 \isFinal}
    \end{gather*}

    \caption{Termination rules.}
    \label{fig:chor-semantics-termination}
  \end{subfigure}

  \begin{subfigure}{\textwidth}
    \begin{gather*}
    \dfrac{}{\nil \trpt[]{\ell} \nil}
    \quad
    \dfrac{c \trpt[]{\ell} c}{c^* \trpt[]{\ell} c^*}
    \quad
    \dfrac{c \not \trpt[]{\ell} c}{c^* \trpt[]{\ell} \nil}
    \quad
    \dfrac{c_1 \trpt[]{\ell} c_1' \quad c_2 \trpt[]{\ell} c_2' \quad \dagger \in \set{{\seq}, {\parallel}, {\alt}}}{c_1 \dagger c_2 \trpt[]{\ell} c_1' \dagger c_2'}
    \\[2mm]
    \dfrac{c_1 \trpt[]{\ell} c_1' \quad c_2 \not \trpt[]{\ell}}{c_1 \alt c_2 \trpt[]{\ell} c_1'}
    \quad
    \dfrac{c_1 \not \trpt[]{\ell} \quad c_2 \trpt[]{\ell} c_2'}{c_1 \alt c_2 \trpt[]{\ell} c_2'}
    \quad
    \dfrac{\subj{\ell} \notin \set{\pt{a}, \pt{b}}}{\cod{a->b:x} \trpt[]{\ell} \cod{a->b:x}}
    \quad
    \dfrac{\subj{\ell} \neq \pt{b}}{\cod{ab?x} \trpt[]{\ell} \cod{ab?x}}
    \end{gather*}

    \caption{Partial termination rules.}
    \label{fig:chor-semantics-partial-termination}
  \end{subfigure}

  \caption{Operational semantics of choreographies.}
  \label{fig:chor-semantics}
\end{figure}

\paragraph{Partial termination}

% Let $\L$ be the set of actions $\set{\cod{ab!x}, \cod{ab?x}}$ for all $\pt{a} \neq \pt{b}$ and $\ms{x}$.
% For all actions the \emph{subject} of the action is its active participant: the subject of a send action $\cod{ab!x}$ is $\pt{a}$ and the subject of the corresponding receive action $\cod{ab?x}$ is $\pt{b}$.
In a weak sequential composition $c_1 \seq c_2$, an action $\ell$ in $c_2$ can be executed if $c_1$ can \emph{partially terminate} for the subject of $\ell$.
Conceptually, a choreography $c_1$ can partially terminate for the subject of $\ell$ by discarding all branches of its behaviour which would conflict with it, i.e., in which the subject of $\ell$ occurs.
This is written $c_1 \trpt{\ell} c_1'$, where $c_1'$ is the remainder of $c_1$ after discarding all branches involving the subject of $\ell$.
For example, if $c_1 = \cod{a->b:x} \alt \cod{a->c:x}$ then $c_1 \trpt{\cod{cd!x}} \cod{a->b:x}$, as this branch does not contain $\pt{c}$.
An exception is when the subject of $\ell$ occurs in \emph{every} branch of $c_1$, in which case $c_1$ cannot partially terminate for the subject of $\ell$, i.e., $c_1 \not \trpt{\ell}$.
In the above example, $c_1 \not \trpt{\cod{ad!x}}$.

The rules for partial termination are deterministic and only discard the absolutely necessary.
In the example above, $c_1 \trpt{\cod{da!x}} c_1$ since the subject $\pt{d}$ does not occur in either branch: dropping one of the branches would be unnecessary and is thus not allowed.
The rules for partial termination are defined in \Cref{fig:chor-semantics-partial-termination}.
We highlight the rules for operators:
\begin{itemize}
  \item Sequential composition $c_1 \seq c_2$ and parallel composition $c_1 \parallel c_2$ can partially terminate if both $c_1$ and $c_2$ can.

  % \pagebreak %Temp

  \item A choice $c_1 \alt c_2$ can partially terminate if at least one of its branches can.
  If both branches can partially terminate then both are kept, otherwise only the partially terminated one is kept.

  \item Following Rensink and Wehrheim, a loop $c^*$ can partially terminate if its body ($c$) can partially terminate without discarding any branches, i.e., if $c \trpt{\ell} c$. In that case also $c^* \trpt{\ell} c^*$.
  Otherwise we allow $c^*$ to be skipped entirely, represented as partial termination to $\nil$, i.e., $c^* \trpt{\ell} \nil$.
  This can happen either if $c$ can partially terminate to $c'$ but $c' \neq c$, or if $c$ cannot partially terminate at all.
  We use $c \not \trpt{\ell} c$ as a shorthand to cover both these cases.
  Skipping a loop is necessary, for example, in a modified master-workers protocol where the master can send an arbitrary number of tasks to the workers, followed by an $\ms{end}$ message to indicate termination.
  With one worker, this protocol is expressed as $(\cod{m->w_1:t} \seq \cod{w_1->m:d})^* \seq \cod{m->w_1:end}$.
  In this choreography, the loop has to eventually partially terminate to $\nil$ to allow for the action $\cod{mw_1!end}$.
  % In particular, if $c = (\cod{a->b:x} \alt \cod{a->c:x})$ then $c^* \not\trpt{\cod{ba!x}} (\cod{a->c:x})^*$ and therefore $c^* \trpt{\cod{ba!x}} \nil$.
  % \jose{maybe mention somwhere that partial termination is deterministic...}
\end{itemize}

% A choreography $c_1$ can partially terminate for $\ell$ if the subject of $\ell$ does not occur in $c_1$ at all.
% Additionally, $c_1$ can also partially terminate for $\ell$ if it can commit to a future course of action in which the subject of $\ell$ does not occur, e.g., by committing to specific branches of choices or by committing to skip a loop.
% We note that these reductions are only applied if necessary: in the above example, $c_1$ is not allowed to commit to a specific branch to partially terminate for an action $\cod{dc!x}$, as --- although the recipient $\pt{c}$ occurs in the second branch --- its subject $\pt{d}$ does not occur in either of its branches.

% If $c_1$ contains a loop, it cannot choose to reduce the loop body, e.g., if $c_1 = (\cod{a->b:x} \alt \cod{a->c:x})^*$ then it cannot commit to only choosing $\cod{a->b:x}$ every iteration (i.e., reduce to $(\cod{a->b:x})^*$) to partially terminate for an action $\cod{cd!x}$.
% However, it can commit to skipping the loop entirely by reducing it to $\nil$.
% Finally, in the case of the composition operators $\seq$ and $\parallel$, a composed choreography can only reduce and partially terminate if both of its branches can.

% If $c$ can partially terminate for $\ell$ we write $c \trpt{\ell} c'$, where $c'$ is either $c$ itself or a (minimally) reduced version of it.
% Formally, the rules for partial termination are defined on the bottom two lines of \Cref{fig:chor-semantics}.

\pagebreak

\begin{example}
  Let $c_1 = (\cod{a->b:x} \alt \cod{a->c:x}) \seq (\cod{d->b:x} \alt \cod{d->e:x})$.
  Let $c_2 = (\cod{a->b:x} \alt \cod{c->b:x})^* \parallel (\cod{c->a:x} \alt \cod{c->b:x})$.

  \begin{itemize}
    \item $c_1 \trpt{\cod{be!x}} \cod{a->c:x} \seq \cod{d->e:x}$.
          The subject $\pt{b}$ of $\cod{be!x}$ occurs in one branch of each of both choices. While the recipient $\pt{e}$ also occurs in the second branch of the second choice, since it is not the actual subject it does not create a conflict.

    \item $c_1 \not \trpt{\cod{ab!x}}$.
          While the second choice can partially terminate without reducing, the first choice contains the subject $\pt{a}$ of $\cod{ab!x}$ in both of its branches. Since one of the choices cannot reduce, neither can their sequential composition.
    
    \item $c_2 \trpt{\cod{ad!x}} \nil \parallel \cod{c->b:x}$.
          The subject $\pt{a}$ of $\cod{ad!x}$ only occurs in one branch of the loop body, but the loop can only reduce to $\nil$.
          On the right hand side of the parallel composition, $\pt{a}$ occurs only in the first branch.

    \item $c_2 \not \trpt{\cod{cd!x}}$.
          While the loop can again reduce to $\nil$, the subject $\pt{c}$ of $\cod{cd!x}$ occurs in both branches of the right hand side of the parallel composition. Since its right hand side cannot partially terminate, neither can it as a whole.

  \end{itemize}
\end{example}

As already discovered by Rensink and Wehrheim~\cite{DBLP:journals/acta/RensinkW01}, an unwanted consequence of these rules for partial termination is that unfolding iterations of loops no longer preserves behaviour.
We would like $c^*$ and $(c \seq c^*) \alt \nil$ to behave the same, but this is not the case.
For example, if $c = \cod{a->b:x} \alt \cod{c->d:x}$, then $c^* \trpt{\cod{ab!x}} \nil$ but $(c \seq c^*) \alt \nil \trpt{\cod{ab!x}} (\cod{c->d:x} \seq \nil) \alt \nil$.
Then $c^* \seq c \tr{\cod{ab!x}} \cod{ab?x}$ by skipping the loop;
however, $((c \seq c^*) \alt \nil) \seq c$ has no way to match this as it can skip the loop but it can only reduce the already unfolded iteration $c$ to $\cod{c->d:x}$ --- it cannot discard it entirely.
We borrow the solution that Rensink and Wehrheim offer, which is the concept \emph{dependent guardedness}.

\paragraph{Dependent guardedness}

A loop $c^*$ is \emph{dependently guarded} if, for all actions $\ell$, the loop body $c$ can only partially terminate for the subject of $\ell$ if it does not occur in $c$ at all.
In other words: any participant that occurs in some branch of $c$ must also occur in every other branch of $c$.
It then follows that $c$ can either partially terminate for the subject of $\ell$ without having to reduce, or it cannot partially terminate at all.
Formally: if $c \trpt{\ell} c'$ then $c' = c$.
A choreography $\hat{c}$ is then dependently guarded if all of its loops are.

As a consequence, we avoid the problem above: if $c^* \trpt{\ell} \nil$ then $c \not\trpt{\ell}$ and $(c \seq c^*) \alt \nil$ is also forced to reduce to the second branch of the choice, which is $\nil$.
More precisely, let $c^*$ be some dependently guarded expression.
If $c \trc{\ell} c'$ for some $\ell, c'$, then $c' = c$.
It follows that $c^* \trc{\ell} c^*$ and $(c \seq c^*) \alt \nil \trc{\ell} (c \seq c^*) \alt \nil$.
Similarly, if $c \not \trc{\ell}$ then $c^* \trc{\ell} \nil$ and $(c \seq c^*) \alt \nil \trc{\ell} \nil$.

\begin{example}
Let $c_1 = \cod{a->b:x} \alt \cod{a->c:x}$.
Let $c_2 = \cod{a->b:x} \alt \cod{b->a:x}$.
% <<<<<<< HEAD
%   \strut
%   \begin{itemize}
%     \item Let $c = \cod{a->b:x} \alt \cod{a->c:x}$. Then $c^*$ is not dependently guarded as $c \trc{\cod{cd!x}} \cod{a->b:x} \neq c$. However, $c$ itself \emph{is} dependently guarded as it does not contain any loop.

%     \item Let $c = \cod{a->b:x} \alt \cod{b->a:x}$. Then $c^*$ is dependently guarded since both $\pt{a}$ and $\pt{b}$ occur in both branches of $c$. However, $(c^*)^*$ is \emph{not} dependently guarded, since $c^* \trc{\cod{ab!x}} \nil$.

%     \item \luc[TODO]{If possible, add a third example.}
%   \end{itemize}
% =======
\begin{itemize}
  \item $c_1^*$ is not dependently guarded as $c_1 \trc{\cod{cd!x}} \cod{a->b:x} \neq c_1$. However, $c_1$ itself \emph{is} dependently guarded as it does not contain any loop.
% \end{example}
% \begin{example}

  \item $c_2^*$ is dependently guarded since both $\pt{a}$ and $\pt{b}$ occur in both branches of $c_2$. However, $(c_2^*)^*$ is \emph{not} dependently guarded, since $c_2^* \trc{\cod{ab!x}} \nil$.
\end{itemize}
\end{example}

\section{\Npoms}
\label{sec:pomsets}

In this section, we formally define the syntax and semantics of \npoms{}.
Additionally, we define a pomset interpretation of expressions in our choreographic language and we show this interpretation to be faithful by proving that it is bisimilar to the original choreography.

A partially ordered multiset~\cite{DBLP:journals/ijpp/Pratt86}, or pomset for short, consists of a set of nodes $E$ (events), a labelling function $\lambda$ to map events to some set of labels (e.g., send and receive actions), and a partial order ${\leq}$ to define dependencies between pairs of events (e.g., an event, or rather its corresponding action, can only fire if all events preceding it in the partial order have already fired).
Its behaviour is the set of all sequences of the labels of its events that abide by ${\leq}$.

For example, for the pomset in \Cref{fig:automaton-vs-pomset}, $E = \set{e_1, \ldots, e_8}$, $\lambda = \set{e_1 \mapsto \cod{mw_1!t}, e_2 \mapsto \cod{mw_1?t}, e_3 \mapsto \cod{w_1m!d}, e_4 \mapsto \cod{w_1m?d}, e_5 \mapsto \cod{mw_2!t}, e_6 \mapsto \cod{mw_2?t}, e_7 \mapsto \cod{w_2m!d}, e_8 \mapsto \cod{w_2m?d}}$, and ${\leq} = \set{(e_i, e_j) \mid (i, j \in [1,4] \lor i, j \in [5,8]) \land i \leq j}$.
Its behaviour consists of all interleavings of $\cod{mw_1!t} \; \cod{mw_1?t} \; \cod{w_1m!d} \; \cod{w_1m?d}$ and $\cod{mw_2!t} \; \cod{mw_2?t} \; \cod{w_2m!d} \; \cod{w_2m?d}$.

% \Cref{fig:automaton-vs-pomset} gives a graphical representation of a pomset.
% Rather than the events themselves only their labels are shown.
% The arrows visualise the partial order: an event precedes any other event to which it has an outgoing arrow, either directly or transitively.
% In this example, the event with label $\cod{mw_1!t}$ precedes the event with label $\cod{mw_1?t}$ directly and the events with labels $\cod{w_1m!d}$ and $\cod{w_1m?d}$ transitively.
% However, it is independent of the events involving $\pt{w_2}$.

% The behaviour represented by a pomset is the set of all its linearisations, i.e., all sequences of the (labels of the) events that respect their partial order. % on them.
% The set of linearisations of the pomset in \Cref{fig:automaton-vs-pomset} consists of all interleavings of the two threads $\cod{mw_1!t} \cod{mw_1?t} \cod{w_1m!d} \cod{w_1m?d}$ and $\cod{mw_2!t} \cod{mw_2?t} \cod{w_2m!d} \cod{w_2m?d}$.
% This explicit concurrency yields a compact representation of the possible interleavings, unlike the automaton in the same figure.

As illustrated in \Cref{fig:pomset-vs-branching-pomset}, however, traditional pomsets suffer from the same problem when representing choices that automata suffer from when representing concurrency: there is no explicit representation of choices in pomsets, and they are represented only implicitly as a set of possible pomsets.
We tackle this by extending pomsets with an explicit representation of choices: a branching structure on events.

\paragraph{Branching structure}
The general idea of a \npom{} is that all possible events are initially part of the pomset, but some are defined as being part of a choice.
To fire these, all relevant choices must first be resolved by replacing the choice with one of its branches, thereby discarding the other branch.
This same idea governs the operational semantics of choreographies (\Cref{fig:chor-semantics}): both branches of a choice are initially part of the choreography but, to proceed in one of them, the other must be dropped.

The branching structure does not interrupt the partial order and all events still participate in it, as shown in \Cref{fig:branching-pomset}, where arrows flow both into and out of the branches of the choice.
As such, a choice can also be resolved to fire an event which is only preceded by one of the branches, reminiscent of the partial termination of choices (\Cref{fig:chor-semantics-partial-termination}).
For example, in \Cref{fig:branching-pomset} the upper branch ($\cod{b->c:x}$) can be discarded to fire the event labelled $\cod{cd!x}$, as it is not dependent on the lower branch.
As shown in \Cref{fig:branching-pomset-nesting}, nested choices are supported as well.

\begin{figure}[t]
  \centering

  \begingroup
  \def\h{1.6}
  \def\hb{1.8}
  \def\v{.4}
  \begin{tikzpicture}[scale = 1, ->, >=stealth, baseline = (anchor)]
    \node (anchor) at (0,0) {};

    \node[event] (ab!) at (0, 0) {$\cod{ab!x}$};
    \node[event] (ab?) at (\h, 0) {$\cod{ab?x}$};

    \draw (ab!) -- (ab?);

    \node[event] (bc!) at (\h + \hb, \v) {$\cod{bc!x}$};
    \node[event] (bc?) at (2*\h + \hb, \v) {$\cod{bc?x}$};
    \node[event] (bd!) at (\h + \hb, -\v) {$\cod{bd!x}$};
    \node[event] (bd?) at (2*\h + \hb, -\v) {$\cod{bd?x}$};

    \draw (bc!) -- (bc?);
    \draw (bd!) -- (bd?);

    \pomsetboxb{bc}{(bc!) (bc?)}
    \pomsetboxb{bd}{(bd!) (bd?)}

    \choiceboxb[2]{b}{(bc) (bd)}

    \node[event] (cd!) at (2*\h + 2*\hb, 0) {$\cod{cd!x}$};
    \node[event] (cd?) at (3*\h + 2*\hb, 0) {$\cod{cd?x}$};

    \draw (cd!) -- (cd?);

    \draw[shorten < = 2pt, shorten > = -1pt] (ab?) |- (bc!);
    \draw[shorten < = 2pt, shorten > = -1pt] (ab?) |- (bd!);
    \draw[shorten > = 2pt] (bc?) -| (cd!);
    \draw[shorten > = 2pt] (bd?) -| (cd?);

    \pomsetbox{(ab!) (cd?) (b)}
  \end{tikzpicture}
  \endgroup

  \caption{A \npom{} representing the choreography $\cod{a->b:x} \seq (\cod{b->c:x} \alt \cod{b->d:x}) \seq \cod{c->d:x}$.}
  \label{fig:branching-pomset}
\end{figure}

\begin{figure}[t]
  \centering

  \begingroup
  \def\h{1.6}
  \def\hb{.2}
  \def\v{1.2}
  \def\vb{.4}
  \begin{tikzpicture}[scale = 1, ->, >=stealth, baseline = (anchor)]
    \node (anchor) at (0,0) {};

    \node[event] (ab!) at (0, \v) {$\cod{ab!x}$};
    \node[event] (ab?) at (\h, \v) {$\cod{ab?x}$};
    \node[event] (ac!) at (0, -\v) {$\cod{ac!x}$};
    \node[event] (ac?) at (\h, -\v) {$\cod{ac?x}$};

    \draw (ab!) -- (ab?);
    \draw (ac!) -- (ac?);

    \node[event] (ba!) at (2*\h + \hb, \v + \vb) {$\cod{ba!x}$};
    \node[event] (ba?) at (3*\h + \hb, \v + \vb) {$\cod{ba?x}$};
    \node[event] (bd!) at (2*\h + \hb, \v + -\vb) {$\cod{bd!x}$};
    \node[event] (bd?) at (3*\h + \hb, \v + -\vb) {$\cod{bd?x}$};

    \draw (ba!) -- (ba?);
    \draw (bd!) -- (bd?);

    \pomsetboxb{ba}{(ba!) (ba?)}
    \pomsetboxb{bd}{(bd!) (bd?)}

    \choiceboxb[2]{b}{(ba) (bd)}

    \draw[shorten < = 2pt, shorten > = -1pt] (ab?) |- (ba!);
    \draw[shorten < = 2pt, shorten > = -1pt] (ab?) |- (bd!);

    \node[event] (ca!) at (2*\h + \hb, -\v + \vb) {$\cod{ca!x}$};
    \node[event] (ca?) at (3*\h + \hb, -\v + \vb) {$\cod{ca?x}$};
    \node[event] (cd!) at (2*\h + \hb, -\v + -\vb) {$\cod{cd!x}$};
    \node[event] (cd?) at (3*\h + \hb, -\v + -\vb) {$\cod{cd?x}$};

    \draw (ca!) -- (ca?);
    \draw (cd!) -- (cd?);

    \pomsetboxb{ca}{(ca!) (ca?)}
    \pomsetboxb{cd}{(cd!) (cd?)}

    \choiceboxb[2]{c}{(ca) (cd)}

    \draw[shorten < = 2pt, shorten > = -1pt] (ac?) |- (ca!);
    \draw[shorten < = 2pt, shorten > = -1pt] (ac?) |- (cd!);

    \pomsetboxb[3]{ab}{(ab!) (b)}
    \pomsetboxb[3]{ac}{(ac!) (c)}
    \choiceboxb[4]{a}{(ab) (ac)}

    \node[event] (da!) at (4*\h + 3*\hb, 0) {$\cod{da!x}$};
    \node[event] (da?) at (5*\h + 3*\hb, 0) {$\cod{da?x}$};

    \draw (da!) -- (da?);
    \draw[shorten > = 2pt] (bd?) -| (da!);
    \draw[shorten > = 2pt] (ba?) -| (da?);
    \draw[shorten > = 2pt] (cd?) -| (da!);
    \draw[shorten > = 2pt] (ca?) -| (da?);

    \pomsetbox{(a) (da?)}
  \end{tikzpicture}
  \endgroup

  \caption{A \npom{} representing the choreography $((\cod{a->b:x} \seq (\cod{b->a:x} \alt \cod{b->d:x})) \alt \allowbreak (\cod{a->c:x} \seq (\cod{c->a:x} \alt \cod{c->d:x}))) \seq \cod{d->a:x}$.}
  \label{fig:branching-pomset-nesting}
\end{figure}

Formally, the branching structure is defined below as a tree with root node $\N$, whose children are either a single event $e$ or a choice node $\C$ with children (branches) $\N_1, \N_2$.
All leaves are events.
\begin{align*}
  \N &::= \set{\C_1, \ldots, \C_n}
  \\
  \C &::= e ~|~ \set{\N_1, \N_2}
\end{align*}

For example, for the pomset in \Cref{fig:branching-pomset}, if $E = \set{e_1, \ldots, e_8}$ and $\lambda = \set{e_1 \mapsto \cod{ab!x}, e_2 \mapsto \cod{ab?x}, e_3 \mapsto \cod{cd!x}, e_4 \mapsto \cod{cd?x}, e_5 \mapsto \cod{bc!x}, e_6 \mapsto \cod{bc?x}, e_7 \mapsto \cod{bd!x}, e_8 \mapsto \cod{bd?x}}$, then its branching structure is $\set{e_1, e_2, e_3, \allowbreak e_4, \set{\set{e_5, e_6}, \set{e_7, e_8}}}$.
For the pomset in \Cref{fig:branching-pomset-nesting}, if $E = \set{e_1, \ldots, e_{14}}$ and $\lambda = \set{e_1 \mapsto \cod{ab!x}, e_2 \mapsto \cod{ab?x}, e_3 \mapsto \cod{ba!x}, e_4 \mapsto \cod{ba?x}, e_5 \mapsto \cod{bd!x}, e_6 \mapsto \cod{bd?x}, e_7 \mapsto \cod{ac!x}, e_8 \mapsto \cod{ac?x}, e_9 \mapsto \cod{ca!x}, e_{10} \mapsto \cod{ca?x}, e_{11} \mapsto \cod{cd!x}, e_{12} \mapsto \cod{cd?x}, e_{13} \mapsto \cod{da!x}, e_{14} \mapsto \cod{da?x}}$, then its branching structure is $\set{e_{13}, e_{14}, \set{\set{e_1, e_2, \set{\set{e_3, e_4}, \allowbreak \set{e_5, e_6}}}, \set{e_7, e_8, \set{\set{e_9, e_{10}}, \set{e_{11}, e_{12}}}}}}$.
By resolving the outer choice and picking its upper branch ($\cod{a->b:x}$), we drop events $e_7, \ldots, e_{12}$ and obtain the middle \npom{} in \Cref{fig:pomset-refine}, with events $e_1, \ldots, e_6, e_{13}, e_{14}$ and branching structure $\set{e_1, e_2, e_{13}, e_{14}, \set{\set{e_3, e_4}, \set{e_5, e_6}}}$.
% and for the one in \Cref{fig:branching-pomset-nesting} is $\set{e_1, e_2, \set{\set{e_3, e_4, \set{\set{e_5, e_6}, \set{e_7, e_8}}}, \set{e_9, e_{10}, \set{\set{e_{11}, e_{12}}, \set{e_{13}, e_{14}}}}}}$.
% We note that, following the structure of our choreographic language, this branching structure only allows binary choices.

We now formally define \npoms{}.

\begin{definition}[\Npom{}]
\label{def:branching-pomset}
  A \npom{} is a four-tuple $R = \tpl{E, {\dd}, \lambda, \N}$, where $E$ is a set of events, ${\dd} \subseteq E \times E$ is such that ${\dd}^\star$ (the transitive closure of ${\dd}$) is a partial order on events, $\lambda : E \mapsto \L$ is a labelling function assigning an action to every event, and $\N$ is a branching structure such that the set of leaves of $\N$ is $E$ and no event in $E$ occurs in $\N$ more than once.
  We use $R.E$, $R.{\dd}$, $R.\lambda$ and $R.\N$ to refer to the components of $R$.
\end{definition}

\paragraph{Semantics}
To fire an event in a \npom{}, on top of being minimal it must also be \emph{active}, i.e., it must not be inside any choice.
In other words: it must be a child of the branching structure's root node.
We thus define a set of refinement rules in \Cref{fig:pom-semantics-refine}, written $R \sqsupseteq R'$, which can be used to resolve choices and move events upwards in the branching structure.

\begin{figure}[t]

    \begin{subfigure}{\textwidth}
      \begin{gather*}
      \dfrac{}{\N \sqsupseteq \N}
      [\text{\rulerefl}]
      \quad
      \dfrac{\N \sqsupseteq \N' \sqsupseteq \N''}{\N \sqsupseteq \N''}
      [\text{\ruletrans}]
      \quad
      \dfrac{
        i \in \set{1, 2}
      }{
        \set{\set{\N_1, \N_2}} \cup \N \sqsupseteq \N_i \cup \N
      }
      [\text{\rulechoice}]
      \\[2mm]
      \dfrac{
        \N_1 \sqsupseteq \N_1' \quad \N_2 \sqsupseteq \N_2'
      }{
        \set{\set{\N_1, \N_2}} \cup \N \sqsupseteq \set{\set{\N_1', \N_2'}} \cup \N
      }
      [\text{\rulecongr}]
      \quad
      \dfrac{R.\N \sqsupseteq \N'}{R \sqsupseteq R[\N']}
    \end{gather*}

    \caption{Refinement rules, where we assume for \rulechoice{} and \rulecongr{} that $\set{\N_1, \N_2} \notin \N$.}
    \label{fig:pom-semantics-refine}
  \end{subfigure}
  
  \begin{subfigure}{\textwidth}
    \begin{gather*}
      \dfrac{R.\N \sqsupseteq \emptyset}{R \isFinal}
      \quad
      \dfrac{
        \begin{gathered}
          R \sqsupseteq R'
          \quad
          e \in \amin(R')
          \\
          \forall R'' : R \sqsupseteq R'' \sqsupset R' \Rightarrow e \notin \amin(R'')
        \end{gathered}
      }{
        R \trc[]{e} R'
      }
      \quad
      \dfrac{R \trc{e} R'}{R \tr[]{e} R' - e}
      \quad
      \dfrac{R \tr{e} R'}{R \tr[]{\lambda(e)} R'}
      \end{gather*}

      \caption{Reduction and termination rules.}
      \label{fig:pom-semantics-reduction}
    \end{subfigure}

  \begin{subfigure}{\textwidth}
    \begin{align*}
      \tpl{E, {\dd}, \lambda, \N}[\N'] &= \tpl{E|_{\N'}, {\dd}|_{\N'}, \lambda|_{\N'}, \N'}
      \\
      X|_{\N} &= \text{restricts $X$ only to the events in $\N$}
      \\
      \amin(R) &= \set{e \in R.E ~|~ \nexists e' \in R.E : e' < e} \land e \in R.\N
      \\
      \hat{e} - e &= \hat{e}
      \\
      \set{\C_1, \ldots, \C_n} - e &=
      \begin{cases}
        \set{\C_1, \ldots, \C_{i - 1}, \C_{i + 1}, \ldots, \C_n} & \text{if $C_i = e$}
        \\
        \set{\C_1 - e, \ldots, \C_n - e} & \text{otherwise}
      \end{cases}
      \\
      \set{\N_1, \N_2} - e &= \set{\N_1 - e, \N_2 - e}
      \\
      R - e &= R[R.\N - e]
    \end{align*}

    \caption{Operations on \npoms{}.}
    \label{fig:pom-semantics-operations}
  \end{subfigure}

  \caption{Semantics of \npoms{}.}
  \label{fig:pom-semantics}
\end{figure}

The first two rules, \rulerefl{} and \ruletrans{}, are straightforward.
The third rule, \rulechoice{}, resolves choices.
It states that we can replace a choice with one of its branches.
This rule serves a dual purpose: by applying it to the outer choice of the pomset in \Cref{fig:branching-pomset-nesting} we can fire the event $\cod{ab!x}$ in its first branch; alternatively, by applying it to the pomset in \Cref{fig:branching-pomset} we can discard one branch of the choice and then fire the event $\cod{cd!x}$, which is now minimal.
The latter use corresponds with the partial termination rules for choreographies.
The fourth rule, \rulecongr{} is used for more fine-grained partial termination.
To make the event $\cod{da!x}$ minimal in \Cref{fig:branching-pomset-nesting} we could resolve two choices with \rulechoice{} (and \ruletrans{}).
However, as the rules for partial termination tell us, it is unnecessary to resolve the outer choice.
Instead, we can apply \rulechoice{} to both inner choices and apply \rulecongr{} to the outer choice to update it without unnecessarily resolving it.
Finally, the fifth rule overloads the refinement notation to also apply to \npoms{} themselves: if $R.\N$ can refine to some $\N'$, then $R$ itself can refine to a derived \npom{} with branching structure $\N'$, whose events are restricted to those occurring in $\N'$ and likewise for ${\leq}$ and $\lambda$.

\begin{figure}[t]
  \centering

  \begingroup
  \def\h{1.6}
  \def\hb{1.8}
  \def\v{.4}
  \begin{tikzpicture}[scale = 1, ->, >=stealth, baseline = (anchor)]
    \node (anchor) at (0,0) {};

    \node[event] (ab!) at (0, 0) {$\cod{ab!x}$};
    \node[event] (ab?) at (\h, 0) {$\cod{ab?x}$};

    \draw (ab!) -- (ab?);

    \node[event] (bd!) at (\h + \hb, -\v) {$\cod{bd!x}$};
    \node[event] (bd?) at (2*\h + \hb, -\v) {$\cod{bd?x}$};

    \draw (bd!) -- (bd?);

    \node[event] (cd!) at (2*\h + 2*\hb, 0) {$\cod{cd!x}$};
    \node[event] (cd?) at (3*\h + 2*\hb, 0) {$\cod{cd?x}$};

    \draw (cd!) -- (cd?);

    \draw[shorten < = 2pt, shorten > = -1pt] (ab?) |- (bd!);
    \draw[shorten > = 2pt] (bd?) -| (cd?);

    \pomsetbox{(ab!) (bd?) (cd?)}
  \end{tikzpicture}
  \endgroup

  \medskip

  Obtained by applying \rulechoice{} to the pomset in \Cref{fig:branching-pomset}.

  \bigskip

  \begingroup
  \def\h{1.6}
  \def\hb{.2}
  \def\v{1.2}
  \def\vb{.4}
  \begin{tikzpicture}[scale = 1, ->, >=stealth, baseline = (anchor)]
    \node (anchor) at (0,0) {};

    \node[event] (ab!) at (0, \v) {$\cod{ab!x}$};
    \node[event] (ab?) at (\h, \v) {$\cod{ab?x}$};

    \draw (ab!) -- (ab?);

    \node[event] (ba!) at (2*\h + \hb, \v + \vb) {$\cod{ba!x}$};
    \node[event] (ba?) at (3*\h + \hb, \v + \vb) {$\cod{ba?x}$};
    \node[event] (bd!) at (2*\h + \hb, \v + -\vb) {$\cod{bd!x}$};
    \node[event] (bd?) at (3*\h + \hb, \v + -\vb) {$\cod{bd?x}$};

    \draw (ba!) -- (ba?);
    \draw (bd!) -- (bd?);

    \pomsetboxb{ba}{(ba!) (ba?)}
    \pomsetboxb{bd}{(bd!) (bd?)}

    \choiceboxb[2]{b}{(ba) (bd)}

    \draw[shorten < = 2pt, shorten > = -1pt] (ab?) |- (ba!);
    \draw[shorten < = 2pt, shorten > = -1pt] (ab?) |- (bd!);

    \node[event] (da!) at (4*\h + 3*\hb, 0) {$\cod{da!x}$};
    \node[event] (da?) at (5*\h + 3*\hb, 0) {$\cod{da?x}$};

    \draw (da!) -- (da?);
    \draw[shorten > = 2pt] (bd?) -| (da!);
    \draw[shorten > = 2pt] (ba?) -| (da?);

    \pomsetbox{(ab!) (da?) (b)}
  \end{tikzpicture}
  \endgroup

  \medskip

  Obtained by applying \rulechoice{} to the outer choice of the pomset in \Cref{fig:branching-pomset-nesting}.

  \bigskip

  \begingroup
  \def\h{1.6}
  \def\hb{.2}
  \def\v{1.2}
  \def\vb{.4}
  \begin{tikzpicture}[scale = 1, ->, >=stealth, baseline = (anchor)]
    \node (anchor) at (0,0) {};

    \node[event] (ab!) at (0, \v) {$\cod{ab!x}$};
    \node[event] (ab?) at (\h, \v) {$\cod{ab?x}$};
    \node[event] (ac!) at (0, -\v) {$\cod{ac!x}$};
    \node[event] (ac?) at (\h, -\v) {$\cod{ac?x}$};

    \draw (ab!) -- (ab?);
    \draw (ac!) -- (ac?);

    \node[event] (ba!) at (2*\h + \hb, \v + \vb) {$\cod{ba!x}$};
    \node[event] (ba?) at (3*\h + \hb, \v + \vb) {$\cod{ba?x}$};
    \phantom{\node[event] (bd!) at (2*\h + \hb, \v + -\vb+0.1) {$\cod{bd!x}$};}
    \phantom{\node[event] (bd?) at (3*\h + \hb, \v + -\vb+0.1) {$\cod{bd?x}$};}

    \draw (ba!) -- (ba?);

    \phantom{\pomsetbox[ba]{(ba!) (ba?)}}
    \phantom{\pomsetbox[bd]{(bd!) (bd?)}}

    \phantom{\choicebox[b]{(ba) (bd)}}

    \draw[shorten < = 2pt, shorten > = -1pt] (ab?) |- (ba!);

    \node[event] (ca!) at (2*\h + \hb, -\v + \vb) {$\cod{ca!x}$};
    \node[event] (ca?) at (3*\h + \hb, -\v + \vb) {$\cod{ca?x}$};
    \phantom{\node[event] (cd!) at (2*\h + \hb, -\v + -\vb+0.1) {$\cod{cd!x}$};}
    \phantom{\node[event] (cd?) at (3*\h + \hb, -\v + -\vb+0.1) {$\cod{cd?x}$};}

    \draw (ca!) -- (ca?);

    \phantom{\pomsetbox[ca]{(ca!) (ca?)}}
    \phantom{\pomsetbox[cd]{(cd!) (cd?)}}

    \phantom{\choicebox[c]{(ca) (cd)}}

    \draw[shorten < = 2pt, shorten > = -1pt] (ac?) |- (ca!);

    \pomsetboxb{ab}{(ab!) (b)}
    \pomsetboxb{ac}{(ac!) (c)}
    \choiceboxb[2]{a}{(ab) (ac)}

    \node[event] (da!) at (4*\h + 3*\hb, 0) {$\cod{da!x}$};
    \node[event] (da?) at (5*\h + 3*\hb, 0) {$\cod{da?x}$};

    \draw (da!) -- (da?);
    \draw[shorten > = 2pt] (ba?) -| (da?);
    \draw[shorten > = 2pt] (ca?) -| (da?);

    \pomsetbox{(a) (da?)}
  \end{tikzpicture}
  \endgroup

  \medskip

  Obtained by applying \rulecongr{} to the outer and \rulechoice{} to both inner choices of the pomset in \Cref{fig:branching-pomset-nesting}.

  \caption{Three refined pomsets.}
  \label{fig:pomset-refine}
\end{figure}

The reduction and termination rules are defined in \Cref{fig:pom-semantics-reduction}.
The first rule simply states that a pomset can terminate if its branching structure can reduce to the empty set.
The second rule defines the conditions for \emph{enabling} an event $e$, written $R \trc{e} R'$.
A \npom{} $R$ can enable $e$ by refining to $R'$ if $e$ is both minimal and active in $R'$ ($e \in \amin(R')$), and if there is no other refinement in between in which $e$ is already minimal and active.
In other words, $R$ may only refine as far as strictly necessary to enable $e$.
This rule implements the same idea as partial termination, with the subtle difference that, whereas partial termination tries to remove any occurrence of a participant, in this case $e$ is actually an event in $R$ itself.
As the two notions are very similar, we use the same notation for enabling events in \npoms{} as for partial termination.
Finally, the last two rules state that, if $R$ can enable $e$ by refining to $R'$, then it can fire $e$ by reducing to $R' - e$, which is the \npom{} obtained by removing $e$ from $R'$ (\Cref{fig:pom-semantics-operations}).
This reduction is defined both on $e$'s label and on the event itself, the latter for internal use in proofs since $\lambda(e)$ is typically not unique but $e$ is.

% \pagebreak %Temp

\begin{example}
  \strut
  \begin{itemize}
    \item $R \trc{e} R'$, where $R$ is the \npom{} in \Cref{fig:branching-pomset}, $R'$ is the topmost \npom{} in \Cref{fig:pomset-refine} and $e$ is the event with label $\cod{cd!x}$.

    \item $R \trc{e} R'$, where $R$ is the \npom{} in \Cref{fig:branching-pomset-nesting}, $R'$ is the middle \npom{} in \Cref{fig:pomset-refine} and $e$ is the event with label $\cod{ab!x}$.

    \item $R \trc{e} R'$, where $R$ is the \npom{} in \Cref{fig:branching-pomset-nesting}, $R'$ is the middle \npom{} in \Cref{fig:pomset-refine} and $e$ is the event with label $\cod{da!x}$.
  \end{itemize}
\end{example}

% As defined in \Cref{fig:pom-semantics-reduction}, a \npom{} can fire an event if it can refine to a pomset in which the event is both minimal and active.
% As with partial termination before, a pomset may only be refined as far as strictly necessary.
% If $R'$ is the maximal refinement of $R$ for which the event $e$ is both minimal and active, written $R \trc{e} R'$, $R$ can evolve to $R' - e$ by firing $e$, written $R \tr{\lambda(e)} R' - e$, where $R' - e$ is the pomset obtained by removing $e$ from $R'$.
% We also include the notation $R \tr{e} R' - e$ for use in proofs, since $\lambda(e)$ is typically not unique.

% {\color{red} Discuss update rules if necessary.}

% \lucin{Problem: take \Cref{fig:branching-pomset}. If we use \rulechoice{} to remove the branch with $\cod{b->c:x}$, the refinement rules in \Cref{fig:pom-semantics} don't remove the transitively derived arrow from, e.g., $\cod{ab!x}$ to $\cod{cd!x}$: the link has been severed, but the pomset does not keep track of that information. Could keep track of two separate sets of dependencies: direct and transitive, and then recompute the transitive set when refining the new direct set (instead of also refining the transitive one).}

% \lucin{It's getting a little big. Either add subsections or move the following to its own section. Leaning towards the latter, paragraph as placeholder.}

\section{\Npoms{} for choreographies}
\label{sec:chor-to-pom}

In this section we formalise the construction of a \npom{} for a choreography $c$ and we show that the pomset semantics for the \npom{} are bisimilar to the operational semantics for $c$.

We have given examples of choreographies and corresponding \npoms{} in \Cref{fig:branching-pomset,fig:branching-pomset-nesting}.
Formally, the rules for the construction of a \npom{} for a choreography $c$, written $\sem{c}$, are defined in \Cref{fig:chor-to-pom}.
Most rules are as expected.
We highlight the rules for operators.
\begin{itemize}
  \item The rule for parallel composition ($\sem{c_1 \parallel c_2}$) takes the pairwise union of all components.

  \item The rule for sequential composition ($\sem{c_1 \seq c_2}$) also adds dependencies to ensure that, for every $\pt{a}$, all events with subject $\pt{a}$ in $\sem{c_1}$ (denoted $E_{1_{\pt{a}}}$) must precede all events with subject $\pt{a}$ in $\sem{c_2}$.
  This matches the reduction rule for weak sequential composition of choreographies (\Cref{fig:chor-semantics-reduction}), as events in $\sem{c_2}$ are only required to wait for events in $\sem{c_1}$ whose subject is the same.

  \item The rule for choice ($\sem{c_1 \alt c_2}$) adds a single top-level choice in the branching structure to choose between the pomsets for $c_1$ and $c_2$.

  \item The rule for loops ($\sem{c^*}$) encodes a loop as a choice between terminating ($\nil$) and unfolding one iteration of the loop ($c \seq c^*$).
  This results in a pomset of infinite size.
  We note that our theoretical results still hold even on infinite pomsets, but that any analysis of an infinite pomset will have to be symbolic.
  However, since the focus of this paper is on supporting choices, we do not discuss this further and leave symbolic analyses for loops for future work.
\end{itemize}

As an example, we construct part of the \npom{} in \Cref{fig:branching-pomset}: $(\cod{b->c:x} \alt \cod{b->d:x}) \seq \cod{c->d:x}$ (thus omitting $\cod{a->b:x}$).
Let $\sem{\cod{b->c:x}} = \tpl{\set{e_1, e_2}, {\leq_1}, \lambda_1, \set{e_1, e_2}}$, $\sem{\cod{b->d:x}} = \tpl{\set{e_3, e_4}, {\leq_2}, \lambda_2, \set{e_3, e_4}}$ and $\sem{\cod{c->d:x}} = \tpl{\set{e_5, e_6}, {\leq_3}, \allowbreak \lambda_3, \set{e_5, e_6}}$ as in \Cref{fig:chor-to-pom}.
First, $\sem{\cod{b->c:x} \alt \cod{b->d:x}} = \tpl{\set{e_1, \ldots, e_4}, {\leq_1} \cup {\leq_2}, \lambda_1 \cup \lambda_2, \set{\set{\set{e_1, e_2}, \set{e_3, e_4}}}}$; this is the pairwise union of the first three components, with the branching structure adding a choice between the two branches.
Then $\sem{(\cod{b->c:x} \alt \cod{b->d:x}) \seq \cod{c->d:x}} = \tpl{\set{e_1, \ldots, e_6}, {\leq_1} \cup {\leq_2} \cup {\leq_3} \cup \set{e_2 \leq e_5, e_4 \leq e_6}, \lambda_1 \cup \lambda_2 \cup \lambda_2, \set{e_5, e_6, \set{\set{e_1, e_2}, \set{e_3, e_4}}}}$; again, this is the pairwise union of all components, with the addition of two dependencies: $e_2 \leq e_5$ represents the arrow in \Cref{fig:branching-pomset} from $\cod{bc?x}$ to $\cod{cd!x}$ as they both have subject $\pt{c}$, $e_4 \leq e_6$ represents the arrow from $\cod{bd?x}$ to $\cod{cd?x}$ as they both have subject $\pt{d}$.
There are no direct dependencies between $e_1$ ($\cod{bc!x}$) or $e_3$ ($\cod{bd!x}$) and either $e_5$ or $e_6$, as the latter two do not have subject $\pt{b}$.

\begin{figure}[t]
  \begin{align*}
    \sem{\nil} &= \tpl{\emptyset, \emptyset, \emptyset, \emptyset}
    \\
    \sem{\cod{a->b:x}} &= \tpl{
      \set{e_1, e_2},
      \set{e_1 \dd e_1, e_1 \dd e_2, e_2 \dd e_2},
      \set{e_1 \mapsto \cod{ab!x}, e_2 \mapsto \cod{ab?x}},
      \set{e_1, e_2}
    }
    \\
    \sem{\cod{ab?x}} &= \tpl{
      \set{e},
      \set{e \dd e},
      \set{e \mapsto \cod{ab?x}},
      \set{e}
    }
    \\
    \sem{c_1 \dagger c_2} &= \sem{c_1} \dagger \sem{c_2} \text{ for } \dagger \in \set{\seq, \alt, \parallel}
    % \sem{c_1 \seq c_2} &= \tpl{
    %   E_1 \cup E_2,
    %   {\dd_1} \cup {\dd_2} \cup {\textstyle \bigcup_{\pt{a} \in \A}} \; E_{1_{\pt{a}}} \times E_{2_{\pt{a}}},
    %   \lambda_1 \cup \lambda_2,
    %   \N_1 \cup \N_2
    % }
    % \\
    % \sem{c_1 \alt c_2} &= \tpl{
    %   E_1 \cup E_2,
    %   {\dd_1} \cup {\dd_2},
    %   \lambda_1 \cup \lambda_2,
    %   \set{\set{\N_1, \N_2}}
    % }
    % \\
    % \sem{c_1 \parallel c_2} &= \tpl{
    %   E_1 \cup E_2,
    %   {\dd_1} \cup {\dd_2},
    %   \lambda_1 \cup \lambda_2,
    %   \N_1 \cup \N_2
    % }
    \\
    \sem{c^*} &= \sem{(c \seq c^*) \alt \nil}
    \\[2mm]
    R_1 \parallel R_2 &= \tpl{
      E_1 \cup E_2,
      {\dd_1} \cup {\dd_2},
      \lambda_1 \cup \lambda_2,
      \N_1 \cup \N_2
    }
    \\
    R_1 \seq R_2 &= \tpl{
      E_1 \cup E_2,
      {\dd_1} \cup {\dd_2} \cup {\textstyle \bigcup_{\pt{a} \in \A}} \; E_{1_{\pt{a}}} \times E_{2_{\pt{a}}},
      \lambda_1 \cup \lambda_2,
      \N_1 \cup \N_2
    }
    \\
    R_1 \alt R_2 &= \tpl{
      E_1 \cup E_2,
      {\dd_1} \cup {\dd_2},
      \lambda_1 \cup \lambda_2,
      \set{\set{\N_1, \N_2}}
    }
  \end{align*}

  \caption{Pomset interpretation of choreographies, where $R_i = \tpl{E_i, {\dd_i}, \lambda_i, \N_i}$ for $i \in \set{1, 2}$, $\A$ is the set of all participants ($\pt{a}, \pt{b}, \ldots$) and $E_{i_{\pt{a}}}$ is the subset of events in $E_i$ with subject $\pt{a}$.}
  \label{fig:chor-to-pom}
\end{figure}

\paragraph{Bisimulation}
For any given a choreography $c$ we can derive two labelled transition systems: 
% We have now defined two labeled transition systems for any choreography $c$:
one from the operational semantics in \Cref{fig:chor-semantics} over $c$, and one from the pomset semantics in \Cref{fig:pom-semantics} over the \npom{} $\sem{c}$ produced by %derived from $c$ by
the rules in \Cref{fig:chor-to-pom}.
In the remainder of this section we show that the two transition systems are bisimilar.

% We have now defined two labeled transition systems for any choreography $c$: one from the operational semantics in \Cref{fig:chor-semantics} on $c$ itself, and one from the pomset semantics in \Cref{fig:pom-semantics} on the \npom{} derived from $c$ by the rules in \Cref{fig:chor-to-pom}.
% We use the remainder of this section to show that the two transition systems are bisimulation equivalent, or \emph{bisimilar}.

Two systems are language equivalent (or trace equivalent) if their languages are the same, i.e., if they accept the same set of words (or traces).
Two systems are bisimilar if each of them can simulate the other, i.e., if they cannot be distinguished from each other just by looking at their behaviour.
This is a stronger notion of equivalence than language equivalence: if two systems are bisimilar then they are also language equivalent, but the inverse is not necessarily true.

% \newpage
\begin{example}
  \strut
  \begin{itemize}
    \item $\cod{a->b:x} \seq (\cod{b->a:x} \alt \cod{b->a:y})$ is language equivalent but not bisimilar to $(\cod{a->b:x} \seq \cod{b->a:x}) \alt (\cod{a->b:x} \seq \cod{b->a:y})$.
    In the former the choice between $\cod{b->a:x}$ and $\cod{b->a:y}$ is made only after $\cod{a->b:x}$, while in the latter the choice is made up front.
    As a result, it is possible in the latter system to fire $\cod{ab!x} \seq \cod{ab?x}$ and then end up in a state where $\cod{ba!x}$ cannot be fired because the branch with $\cod{b->a:y}$ was chosen --- or the other way around; in the former system it is always possible to fire both $\cod{ba!x}$ and $\cod{ba!y}$.

    \item $\cod{a->b:x}$ is bisimilar to $\cod{a->b:x} \alt \cod{a->b:x}$.
    While the latter contains a choice, the two systems cannot be distinguished by their behaviour.
    In both cases, the only allowed action is $\cod{ab!x}$ and then $\cod{ab?x}$.
  \end{itemize}
\end{example}

Formally, two transition systems $A_1, A_2$ are bisimilar, written $A_1 \sim A_2$, if there exists a bisimulation relation $\bisim$ relating the states of $A_1$ and $A_2$ which relates their initial states~\cite{sangiorgi2011introduction}.
The relation $\bisim$ is a bisimulation relation if, for every pair of states $\tpl{p, q} \in \bisim$:
\begin{itemize}
  \item If $p \tr{\ell} p'$ then $q \tr{\ell} q'$ and $\tpl{p', q'} \in \bisim$ for some $q'$, and vice-versa.

  \item If $p \isFinal$ then $q \isFinal$, and vice-versa.
\end{itemize}
In other words: if one of the two can perform a step, then the other can perform a matching step such that the resulting states are again in the bisimulation relation.

This is also the approach we follow when proving that $c \sim \sem{c}$ for all (dependently guarded) choreographies $c$: we define a relation $\bisim = \set{\tpl{c, \sem{c}} \mid \text{$c$ is a dependently guarded choreography}}$ relating all dependently guarded choreographies with their interpretation as \npom{} by the rules in \Cref{fig:chor-to-pom}.
We then show that:
\begin{itemize}
  \item If $c \tr{\ell} c'$ then $\sem{c} \tr{\ell} \sem{c'}$ (\Cref{lem:if-chor-step-then-pom-step}).

  \item If $\sem{c} \tr{\ell} R'$ then $c \tr{\ell} c'$ such that $R' = \sem{c'}$ (\Cref{lem:if-pom-step-then-chor-step}).

  \item If $c \isFinal$ then $\sem{c} \isFinal$ (\Cref{lem:if-chor-term-then-pom-term}).

  \item If $\sem{c} \isFinal$ then $c \isFinal$ (\Cref{lem:if-pom-term-then-chor-term}).
\end{itemize}

Together these lemmas prove that $c \sim \sem{c}$ for all dependently guarded $c$ (\Cref{thm:chor-bisim-pom}).
Most of the proofs are straightforward by structural induction on $c$.
Of particular interest, however, are the two reduction lemmas in the case of weak sequential composition, i.e., if $c_1 \seq c_2 \tr{\ell} c_1' \seq c_2'$ in \Cref{lem:if-chor-step-then-pom-step} and if $\sem{c_1 \seq c_2} \tr{e} R'$ where $e$ is an event in $\sem{c_2}$ in \Cref{lem:if-pom-step-then-chor-step}.
To prove these specific cases we need to show a correspondence between partial termination and enabling events.
We do this with \Cref{lem:max-refinement}, in which we show two directions simultaneously. If the choreography $c_1$ can partially terminate for the subject of an action $\ell$ in $c_2$ then the \npom{} $\sem{c_1 \seq c_2}$ can enable the corresponding event. Conversely, if $\sem{c_1 \seq c_2}$ can enable some event in $\sem{c_2}$ then the choreography $c_1$ can partially terminate for the subject of its label.
When proving these cases in \Cref{lem:if-chor-step-then-pom-step,lem:if-pom-step-then-chor-step}, we then only have to show that the preconditions of \Cref{lem:max-refinement} hold.

In the following, a number of technical lemmas and most of the proofs are omitted in favour of informal proof sketches or highlights.
The omitted proofs and technical lemmas can be found in a separate technical report~\cite{techreport}.

% \begin{definition}
%   \luc[$\bisim = \set{\tpl{c, \sem{c}} \mid \text{$c$ is a choreography}}$.]{Probably doesn't need to be in a separate environment, can just be inline. Or possibly even introduced only all the way in the final proof, as all the formalisations can be done without it.}
% \end{definition}

\begin{restatable}{lemma}{lemmaxrefinement}
\label{lem:max-refinement}
  Let $c_1$ and $c_2$ be dependently guarded choreographies.
  Let $c_2 \tr{\ell} c_2'$ and $\sem{c_2} \trc{e} R_2'$ such that $\lambda(e) = \ell$ and $\sem{c_2'} = R_2' - e$.
  \begin{enumerate}[label=(\alph*)]
    \item If $c_1 \trpt{\ell} c_1'$ then $\sem{c_1 \seq c_2} \trc{e} \sem{c_1'} \seq R_2'$.
    \item If $\sem{c_1 \seq c_2} \trc{e} R_1' \seq R_2'$ then $c_1 \trpt{\lambda(e)} c_1'$ and $\sem{c_1'} = R_1'$.
  \end{enumerate}
\end{restatable}
\begin{proof}[Proof sketch]
  This proof is by structural induction on $c_1$.
  Although the details require careful consideration, it is conceptually straightforward:
  every case in (a) consists of showing that $e$ is minimal and active in $\sem{c_1'} \seq R_2'$ and that $\sem{c_1'} \seq R_2'$ is the first refinement for which this is true, and then applying the second rule in \Cref{fig:pom-semantics-reduction};
  every case in (b) consists of showing that $\sem{c_3 \seq c_2} \trc{e} \sem{c_3'} \seq R_2'$ for some subexpression $c_3$ of $c_1$ and similarly for $c_4$ (e.g., when $c_1 = c_3 \alt c_4$), then applying the induction hypothesis (b) to obtain $c_3 \trpt{\ell} c_3'$ and $c_4 \trpt{\ell} c_4'$, and finally applying the partial termination rules in \Cref{fig:chor-semantics-partial-termination}.
\end{proof}

\begin{restatable}{lemma}{lemifchorstepthenpomstep}
\label{lem:if-chor-step-then-pom-step}
  Let $c$ be a dependently guarded choreography.
  If $c \tr{\ell} c'$ then $\sem{c} \tr{\ell} \sem{c'}$.
\end{restatable}
\begin{proof}[Proof sketch]
  This proof is by structural induction on $c$.
  We note that, if $c = c_1 \seq c_2$ and $c' = c_1' \seq c_2'$, i.e., when partial termination is applied, then the premises of \Cref{lem:max-refinement} hold by the induction hypothesis and the result swiftly follows.
  All other cases are straightforward.
\end{proof}

\begin{restatable}{lemma}{lemifpomstepthenchorstep}
\label{lem:if-pom-step-then-chor-step}
  Let $c$ be a dependently guarded choreography.
  If $\sem{c} \tr{\ell} R'$ for some $R'$ then $c \tr{\ell} c'$ such that $R' = \sem{c'}$.
\end{restatable}
\begin{proof}[Proof sketch]
  This proof is by structural induction on $c$.
  We highlight two cases:
  \begin{itemize}
    \item If $c = c_1^*$ then we use a technical lemma to show that $R' = R_1' \seq \sem{c_1^*}$ such that $\sem{c_1} \tr{\ell} R_1'$.
    It then follows from the induction hypothesis that $c_1 \tr{\ell} c_1'$ such that $\sem{c_1'} = R_1'$.
    The remainder is straightforward.

    \item If $c = c_1 \seq c_2$ then $\sem{c} = \sem{c_1} \seq \sem{c_2}$.
    % We then distinguish two cases.
    % If $e$ is an event in $\sem{c_1}$ then we proceed to show that $\sem{c_1} \tr{\ell} R_1'$ for some $R_1'$, at which point we can apply the induction hypothesis.
    If $e$ is an event in $\sem{c_2}$ then we proceed to show that $\sem{c_2} \tr{\ell} R_2'$, at which point we can apply the induction hypothesis.
    We have then satisfied the premises of \Cref{lem:max-refinement}.
    The remainder is straightforward.
  \end{itemize}
  All other cases are straightforward.
\end{proof}

\begin{restatable}{lemma}{lemifchortermthenpomterm}
\label{lem:if-chor-term-then-pom-term}
  Let $c$ be a dependently guarded choreography.
  If $c \isFinal$ then $\sem{c} \isFinal$.
\end{restatable}
\begin{proof}[Proof sketch]
  This proof is by structural induction on $c$.
  All cases are straightforward.
\end{proof}

\begin{restatable}{lemma}{lemifpomtermthenchorterm}
\label{lem:if-pom-term-then-chor-term}
  Let $c$ be a dependently guarded choreography.
  If $\sem{c} \isFinal$ then $c \isFinal$.
\end{restatable}
\begin{proof}[Proof sketch]
  This proof is by structural induction on $c$.
  All cases are straightforward.
\end{proof}

\begin{restatable}{theorem}{thmchorbisimpom}
\label{thm:chor-bisim-pom}
  Let $c$ be a dependently guarded choreography.
  Then $c \sim \sem{c}$.
\end{restatable}
\begin{proof}
  Recall the relation $\bisim = \set{\tpl{c, \sem{c}} \mid \text{$c$ is a dependently guarded choreography}}$.
  Let $\tpl{c, R} \in \bisim$.
  \begin{itemize}
    \item If $c \tr{\ell} c'$ then $R \tr{\ell} R'$ and $\tpl{c', R'} \in \bisim$ (\Cref{lem:if-chor-step-then-pom-step}).

    \item If $R \tr{\ell} R'$ then $c \tr{\ell} c'$ and $\tpl{c', R'} \in \bisim$ (\Cref{lem:if-pom-step-then-chor-step}).

    \item If $c \isFinal$ then $R \isFinal$ (\Cref{lem:if-chor-term-then-pom-term}).

    \item If $R \isFinal$ then $c \isFinal$ (\Cref{lem:if-pom-term-then-chor-term}).
  \end{itemize}
  Then $\bisim$ is a bisimulation relation and $c \sim \sem{c}$ (\cite{sangiorgi2011introduction}).
\end{proof}

% \begin{itemize}
  % \item Syntax and semantics (and naturally a good explanation of the general idea)

  % \item Translation from choreographic expression to \npom{}
  % \begin{itemize}
  %   \item Resulting \npom{} may be infinitely large $\rightarrow$ discuss implications. Theoretical results still hold, but analysis in practice cannot always be explicit.
  % \end{itemize}

%   \item Behavioural correspondence between choreographic expressions and \npoms{}
% \end{itemize}

% \section{Decomposition and realisability}
% \label{sec:realisability}

% \begin{itemize}
%   \item Projection and composition with network semantics
%   \begin{itemize}
%     \item Our buffers are order-preserving $\rightarrow$ discuss
%     \begin{itemize}
%       \item Necessary to make loops useful

%       \item However, makes $a \rightarrow b:x \parallel a \rightarrow b:y$ unrealisable since the decomposition is more rigid. That being said, the traces of the decomposition still form a subset of those of the global one, which is nice; the reverse is just not necessarily true.
%     \end{itemize}
%   \end{itemize}

%   \item Well-formedness and well-matchedness conditions
%   \begin{itemize}
%     \item Well-formedness almost straight from MPST literature

%     \item Well-matchedness sort of matches the features of a choreography, but applies to \npoms{} in general. Note that \npoms{} derived from choreographies are always well-matched.
%   \end{itemize}

%   \item Behavioural correspondence between \npoms{} and its decomposition
% \end{itemize}

\section{Conclusion}
\label{sec:conclusion}

We have defined a choreography language and its operational semantics (\Cref{fig:chor-syntax,fig:chor-semantics}) using the weak sequential composition and partial termination of Rensink and Wehrheim~\cite{DBLP:journals/acta/RensinkW01}, which is novel in the context of choreographies.
We have defined a model, \npoms{} (\Cref{def:branching-pomset}), which can compactly represent both concurrency and choices, and have defined its semantics (\Cref{fig:pom-semantics}).
We have shown that we can use \npoms{} to model choreographies (\Cref{fig:chor-to-pom}) and that this model is behaviourally equivalent to the operational semantics (\Cref{thm:chor-bisim-pom}).

We believe that \npoms{} can be further improved.
We mention three points in particular and then discuss related work.

\paragraph{Binary choices}

Our branching structure $\N$ only supports binary choices.
This matches the structure of choreographies, but it would be more natural to represent $c_1 \alt (c_2 \alt c_3)$ as a single choice between the pomsets $\sem{c_1}$, $\sem{c_2}$ and $\sem{c_3}$ instead of as two nested binary choices.
However, supporting arbitrary $n$-ary choices also requires some thought about how to change the rules for refinement (\Cref{fig:pom-semantics-refine}), in particular \rulechoice{}.
A naive change would be to simply have this rule use $i \in \set{1, \ldots, n}$ and $\set{\set{\N_1, \ldots, \N_n}}$ instead of its current binary rules, but this is not sufficient as this naive $n$-ary choice would not be equivalent to the same branches composed as nested binary choices.
For example, $c_1 \alt (c_2 \alt c_3)$ can partially terminate to $c_1 \alt c_2$ and its interpretation as a \npom{} can refine to $\sem{c_1 \alt c_2}$, but a \npom{} whose branching structure consists of a single ternary choice $\set{\set{\N_1, \N_2, \N_3}}$ would not be able to refine to $\set{\set{\N_1, \N_2}}$ as the rules would only allow it to refine all of its branches or discard all but one of them.
Properly supporting $n$-ary choices would thus also require a new rule that allows $\set{\set{\N_1, \ldots, \N_m}}$ to refine to choice between an arbitrary (non-empty) subset of its branches.

\paragraph{Partial order}

In \Cref{def:branching-pomset}, ${\leq}$ is defined as a relation on events such that its transitive closure is a partial order, rather than ${\leq}$ being a partial order itself as it is in traditional pomsets.
The need for this change arises from the update rule $R[\N]$ (\Cref{fig:pom-semantics-operations}) in our use case as choreographies.
Consider the \npom{} in \Cref{fig:branching-pomset}.
To match the operational semantics, we should be able to refine this pomset by discarding the $\cod{b->c:x}$ branch of the choice, after which $\cod{cd!x}$ should be minimal.
In our current rules the events $\cod{bc!x}$ and $\cod{bc?x}$ are removed along with their entries in ${\leq}$ and then $\cod{cd!x}$ is minimal.
However, if ${\leq}$ is a partial order, then since a partial order is transitive ${\leq}$ would also contain the entries $\cod{ab!x} \leq \cod{cd!x}$ and $\cod{ab?x} \leq \cod{cd!x}$ and, since these entries do not contain $\cod{bc!x}$ or $\cod{bc?x}$ but are obtained by transitivity, they are not removed.
Consequently, there would be no refinement that enables $\cod{cd!x}$.

In general, if $R_1 \sqsupseteq R_1'$ and $R_2 \sqsupseteq R_2'$ then it would not necessarily be true that $R_1 \seq R_2 \sqsupseteq R_1' \seq R_2'$, as $R_1 \seq R_2$ may contain dependencies obtained by transitivity which would still be present in its updated version but which cannot be derived in $R_1' \seq R_2'$.
We have no ready alternative.
In the case of choreographies it may suffice to provide a more sophisticated update rule which properly trims these unwanted dependencies, but since this relies on knowledge of how these dependencies were derived from choreographies it is difficult to see how this could be applied to \npoms{} in general.

\paragraph{Loops}

In \Cref{fig:chor-to-pom} a loop $c^*$ is encoded by infinitely unfolding it.
As such, \npoms{} do not currently provide a finite representation of infinite choreographies.
This remains a topic for future work, for which we envision two possible directions.
One possibility would be to add an explicit repetition construct to the branching structure (e.g., change the second grammatical rule to $\C = e ~|~ \set{\N_1, \N_2} ~|~ \N^*$) and expand the semantics and proofs accordingly.
Another possibility would be to explore the approach used in message sequence chart graphs~\cite{DBLP:journals/tcs/AlurEY05} and add a graph structure on top of the branching structure.

\paragraph{Related work}
Choreographies are typically used in a top-down workflow: the developer writes a
global view $C$ and decomposes it into its projections, such that the behaviour
of $C$ is \emph{behaviourally equivalent} to the parallel composition of its
projections. %This condition is known as \emph{realisability} of $C$.
Examples of
this approach include workflows based on message sequence
charts~\cite{ITU-T:MSC,DBLP:journals/tcs/AlurEY05}, multiparty session
types~\cite{DBLP:conf/popl/HondaYC08,DBLP:journals/jacm/HondaYC16}, and
choreographic
programs~\cite{DBLP:conf/popl/CarboneM13,DBLP:journals/tcs/Cruz-FilipeM20}.
The choreographic language used in this paper assumes asynchronous communication between agents and includes a finite loop operator, borrowing from this literature the same notion of actions as interactions and their (parallel, sequential, and choice) composition.

Pomsets were initially introduced by Pratt~\cite{DBLP:journals/ijpp/Pratt86} for concurrent models and have been widely used, e.g., in the context of message sequence charts by Katoen and Lambert~\cite{DBLP:conf/fbt/KatoenL98}.
Recently Guanciale and Tuosto proposed two semantic frameworks for choreographies, one of which uses sets of pomsets~\cite{DBLP:journals/jlp/TuostoG18}.
They also note that the pomset framework exhibits exponential growth in the number of choices in a choreography, and they propose an alternative semantic framework using hypergraphs, which can compactly represent choices.
While the hypergraph framework is more compact, their pomset framework is simpler and, they believe, more elegant.
We agree with this analysis, and we aim to preserve the simplicity and elegance of the pomset framework by proposing a framework that avoids exponential growth in the number of choices while still being based on pomsets.
In another recent paper Guanciale and Tuosto use pomsets to reason over choreography realisability~\cite{DBLP:journals/jlap/GuancialeT19}.
This demonstrates the potential of using pomsets for semantic analysis, and we are investigating how to use our framework for similar analysis.

Other related work includes the usage of event structures in the context of binary session types by Castellan and Yoshida~\cite{DBLP:journals/pacmpl/CastellanY19} and multiparty by Castellani et al.~\cite{DBLP:conf/birthday/CastellaniDG19}.
Event structures and \npoms{} both feature a set of events with a causality relation and a choice mechanism.
The main difference between the two approaches is in the latter.
Choices in event structures are based on a conflict relation on events, where two events in conflict cannot occur together in an execution and one of the two must be chosen.
% These are event-level choices: individual pairs of events are marked as conflicts, but the different branches of execution are only implicitly present.
% In contrast, \npoms{} feature higher-level choices, where events are grouped and represented explicitly as branches in the branching structure.
In contrast, we structure events in \npoms{} hierarchically.
In the latter work, Castellani et al. use prime event structures to encode global types. The conflict heredity property of prime event structures leads to event duplication, which is avoided in flow event structures and \npoms{}. We leave a detailed comparison between flow event structures and \npoms{} for future work.
We note that, given a \npom{}, one may construct an event structure by defining its conflict relation as all pairs of events that belong to different branches of some choice in the branching structure.

\subsubsection*{Acknowledgments}
\emph{Funding by G. Cledou and J. Proença:} European Regional Development Fund (ERDF), Operational Programme for Competitiveness and Internationalisation (COMPETE 2020), and Fundação para a Ciência e a Tecnologia (FCT): POCI-01-0145-FEDER-029946 (DaVinci).
\emph{Funding by S. Jongmans:} Netherlands Organisation of Scientific Research: 016.Veni.192.103.
\emph{Funding by J. Proença:} %European Regional Development Fund (ERDF), Operational Programme for Competitiveness and Internationalisation (COMPETE 2020): POCI-01-0145-FEDER- 029946 (DaVinci).
FCT, within the CISTER Research Unit:
UIDB/04234/2020.
% UIDP/UIDB/04234/2020. 
FCT, I.P. (Portuguese Foundation for Science and Technology): PTDC/CCI-COM/4280/2021 (IBEX).
% European Regional Development Fund (ERDF)
ERDF and FCT, Portugal 2020 Partnership Agreement, Norte Portugal Regional Operational Programme (NORTE 2020): NORTE-01-0145-FEDER-028550 (REASSURE).
ECSEL Joint Undertaking (JU): grant agreement No 876852 (VALU3S).

\bibliographystyle{eptcs}
\bibliography{ice2022,bib.bib}

% \clearpage
% \appendix
% \input{proofs}

\end{document}